\documentclass[journal]{IEEEtran}
\usepackage{amsmath,amsfonts}
\usepackage{algorithmic}
\usepackage{algorithm}
\usepackage{array}
\usepackage[caption=false,font=normalsize,labelfont=sf,textfont=sf]{subfig}
\usepackage{textcomp}
\usepackage{stfloats}
\usepackage{url}
\usepackage{verbatim}
\usepackage{graphicx}
\usepackage{cite}
\usepackage{tabularx}
\usepackage{pifont}
\newtheorem{lemma}{Lemma}
\newtheorem{definition}{Definition}
\newtheorem{theorem}{Theorem}
\newtheorem{proof}{Proof}
\begin{document}

\title{Inner-product Functional Encryption with Fine-grained Revocation for Flexible EHR Sharing}

\author{Yue Han, Jinguang Han,~\IEEEmembership{Senior Member,~IEEE,} Liqun Chen,~\IEEEmembership{Senior Member,~IEEE,} and Chao Sun
\thanks{Yue Han and Chao Sun are with the School of Cyber Science and Engineering, Southeast University, Nanjing 210096, China (e-mail: yuehan@seu.edu.cn; sunchaomt@seu.edu.cn).}
\thanks{Jinguang Han is with the School of Cyber Science and Engineering,
Southeast University, Nanjing 210096, China, and also with the Engineering
Research Center of Blockchain Application, Supervision and Management
(Southeast University), Minister of Education, Nanjing 210096, China (e-mail: jghan@seu.edu.cn). }
\thanks{Liqun Chen is with the Department of Computer Science, University of Surrey, Guildford, 
Surrey GU27XH, United Kingdom (e-mail: liqun.chen@surrey.ac.uk)}
}




\maketitle

\begin{abstract}
E-health record (EHR) contains a vast amount of continuously growing medical data and enables medical institutions to access patient health data conveniently. 
This provides opportunities for medical data mining which has important applications in identifying high-risk patients and improving disease diagnosis, etc.
Since EHR contains sensitive patient information, how to protect patient privacy and enable mining on EHR data is important and challenging. 
Traditional public key encryption (PKE) can protect patient privacy, but cannot support flexible selective computation on encrypted EHR data.
Functional encryption (FE) allows authorised users to compute function values of encrypted data without releasing other information, hence supporting selective computation on encrypted data.
Nevertheless, existing FE schemes do not support fine-grained revocation and update, so they are unsuitable for EHR system.
In this paper, to support fine-grained revocation and update,  
we first propose an inner-product functional encryption with fine-grained 
revocation (IPFE-FR) scheme, and then apply it to a flexible EHR sharing system. 
Our scheme possesses the following features: 
(1) a group manager can revoke a specific function computation of medical institutions on encrypted EHR data,  
instead of all function computation rights.
(2) a revoked medical institution is not allowed to compute the function value of encrypted EHR data not only generated after the revocation, but also generated before the revocation.
(3) secret keys issued to the same medical institution are bound together to prevent collusion attacks. 
The formal definition and security model of the IPFE-FR scheme are proposed. 
Furthermore, we present a concrete construction and reduce its security to the Learning with Errors (LWE) assumption which is quantum-resistant. 
Finally, the theoretical analysis and experimental implementation of our scheme are conducted to show its efficiency.  
Our results show that the proposed scheme can flexibly and securely share EHR data.
\end{abstract}


\section{Introduction}
\IEEEPARstart{E}HR is a digital medical information management tool used to integrate, store, manage, and share patients' health data.
Compared with traditional health data management methods, 
EHR improves the efficiency of medical services and the utilisation of health data.
At present, some EHR systems were developed, such as Oracle Health eService \cite{OracleHealth} and IBM Watson Health \cite{IBMWatson}.
To ease maintenance burdens and facilitate data sharing, EHR data is always
outsourced to remote cloud servers.
Using these massive amounts of data stored in cloud servers, medical institutions can employ data mining techniques to identify undiagnosed patients, track infections, etc.
However, with the growing concern over privacy issues, many legal frameworks, such as Health Insurance Portability and Accountability Act (HIPAA) \cite{HIPAA} and General Data Protection Regulation (GDPR) \cite{GDPR}, have come into effect.
Since EHR inherently contain a large amount of patient privacy information, how to achieve selective computation of EHR data while ensuring its confidentiality is an urgent issue that needs to be addressed, which is particularly important in EHR data mining. 

To support fine-grained access control on EHR data, attribute-based encryption (ABE) was integrated with EHR systems \cite{pussewalage2022delegatable,wang2021medshare,li2021ehrchain,zaghloul2020d}. 
However, the decryption mode of ABE is all-or-nothing, namely, EHR data users either fully decrypt the data or cannot know anything about the data. This prevents the execution of
selective computations on encrypted EHR data. 

As a new promising paradigm of public key encryption (PKE), functional encryption (FE) \cite{boneh2011functional} was introduced to enable authorised users to compute a specific function of encrypted data, without disclosing any other information.
By leveraging functional encryption, medical institutions can conduct data analysis without disclosing EHR data. This effectively protects patients' privacy while supporting selective computation of EHR data. 

Inner product plays an important role in the analysis of gene expression data. The weighted mean is a useful tool to describe the characteristics of gene data \cite{kadota2008weighted}. 
Inner-product FE (IPFE) is a specific form of FE and enables an authorised user to compute the inner product of the vectors embedded in his/she secret key and the ciphertext, respectively, without releasing any other information about the encrypted data.
With IPFE, an authorised medical institution
can calculate the weighted means of gene data
without revealing sensitive information, thus protecting patient privacy.
However, there are still some important issues that have not been well considered.

(1) Fine-grained revocation. Consider the following scenario: 
In an EHR sharing system, suppose that a medical institution $MI_1$ has the computing rights of functions $f_1$, $f_2$ and $f_3$. 
When the work assignment of $MI_1$ changes, its partial computing rights of function (e.g., $f_1$ and $f_2$) should be revoked.
In this case, the system needs to revoke $MI_1$'s function computing rights of $f_1$ and $f_2$ while retaining $MI_1$'s function computing right of $f_3$. Without this fine-grained control over medical institution function computing rights, EHR data could be misused or its usability compromised.

(2) Forward security. Moreover, in an EHR data sharing system, forward security must be ensured: revoked medical institutions cannot calculate the function values of encrypted data generated before revocation. Otherwise, medical institutions that have been revoked can still calculate the function values of much encrypted data, which would compromise patient privacy.

(3) Collusion attacks. Furthermore, collusion attacks, where unauthorised users may collude to compute the function values of encrypted EHR data, must be considered. If an EHR data sharing system is not collusion-resistant, then revoked medical institutions may combine their keys to illegally calculate the function value of encrypted EHR data. 


In this paper, to meet the
security and function requirements in EHR sharing schemes, we first propose IPFE with fine-grained revocation (IPFE-FR).
Our scheme has the following features: 
(1) the revocation of partial function computing rights of medical institutions is supported, namely, fine-grained revocation;
(2) encrypted EHR data update and functional key update are supported, so our scheme achieves forward security. If a medical institution's function computing right is revoked, it cannot compute the
function of priorly encrypted EHR data.
Moreover, updating encrypted EHR data is performed by a cloud server to alleviate the burden on CA;
(3) a medical
institution's secret keys are tied together to resist collusion attacks; 
(4) our scheme supports indirect revocation, namely, the revocation is conducted by a group manager instead of EHR holders, which is suitable for data sharing schemes. 

\subsection{Contributions}
Our main contributions are as follows: 
(1) We present a new IPFE scheme for flexible EHR sharing, named IPFE-FR, and formalise its formal definition and security model; 
(2)We construct a concrete IPFE-FR scheme and prove that its security is based on the well-established Learning With Errors (LWE) assumption, which can resist quantum computing attacks;
(3)We theoretically analyse the performance of the concrete IPFE-FR scheme. 
Furthermore, we conduct experimental implementations of the concrete IPFE-FR scheme to show its efficiency. 

\subsection{Challenges and Techniques}
$\mathit{Challenges:}$ 
The following technical challenges arose during the construction of the IPFE-FR scheme.
\begin{enumerate}
    \item 
    How to realise the fine-grained revocation without re-initialising the EHR sharing system, 
    namely, there is no need to re-generate the master secret-public key pair for each revocation.
    \item When revocation occurs, how to update ciphertexts to 
    ensure forward security. 
    Particularly, we want to outsource the operation of updating ciphertexts to a semi-trusted server 
    to reduce the burden on the central authority. 
    In this process, data confidentiality must be guaranteed. 
    \item How to update the key of a specified function for unrevoked users.      

\end{enumerate}
$\mathit{Techniques:}$
The following techniques were employed to address the aforementioned challenges.
\begin{enumerate}
    \item To address the first challenge, we introduce system version numbers into our scheme.  
    Specifically, the central authority generates a master secret-public key pair and initialises the system. 
    Subsequently, the group manager generates a group secret-public key pair associated with the system version number. 
    Secret keys issued to authorised users are associated with system version numbers. 
    When encrypting data, the group public key is applied, ensuring that inner product calculations can be performed only 
    when the system version number embedded in the ciphertext matches that in the user's secret key. 
    Upon the occurrence of a revocation, the group manager updates both the system version number and the group secret-public key pair.  
    Notably, the master secret-public key pair maintained by the central authority does not need to be updated. 
    \item
    To resolve the second challenge, 
    we extend the idea of PRE. 
    The central authority generates update keys for a cloud server. 
    Using the update keys, the cloud server can update the ciphertexts to a new version 
    without compromising the confidentiality of encrypted data.
    Specifically, we employ the IPFE proxy re-encryption (IPFEPRE) technique proposed in \cite{luo2024public}.
    \item To overcome the third challenge, 
    broadcast encryption \cite{agrawal2017efficient} is applied to broadcast the updating information of 
    a specified function to all unrevoked users.
\end{enumerate}

\subsection{Related work} 
Functional encryption (FE) enables selective computation over encrypted data, 
addressing the limitations of traditional PKE which only supports all-or-nothing decryption.
In FE schemes, a function is embedded into a private key, 
so that decryption yields only the corresponding function value of the encrypted data, 
without revealing any additional information.
Boneh et al. \cite{boneh2011functional} first proposed the formal definition and security model of FE. 
At present, existing FE schemes are mainly divided into two types \cite{abdalla2020inner}: 
(1) feasible construction for general functionalities and 
(2) efficient and concrete construction for specific functions, i.e., inner-product, quadratic functions, etc. 

In order to efficiently perform inner-product computations on encrypted data, Abdalla et al. \cite{abdalla2015simple} proposed a selectively secure IPFE scheme. 
In this scheme, the secret key corresponding to vector $\mathbf{x}$ allows the user to obtain only the inner product with the encrypted vector $\mathbf{y}$, that is, $\langle \mathbf{x}, \mathbf{y} \rangle$, without learning anything else.
To enhance the security guarantees of this scheme, Agrawal et al. \cite{agrawal2016fully} introduced an adaptively secure IPFE scheme.
To guarantee that user inputs are valid, Nguyen et al. \cite{nguyen2023verifiable} presented a verifiable IPFE scheme.
To trace malicious users who disclose or trade their secret keys, Do et al. \cite{do2020traceable} introduced a traceable IPFE scheme, where each secret key is bound to a user's identity, enabling black-box traceability.
Qiu et al. \cite{qiu2025privacy} proposed a privacy-preserving and traceable FEIP scheme where privacy and accountability are balanced. 
By integrating IPFE with ABE, Abdalla et al. \cite{abdalla2020inner} designed IPFE schemes that support attribute-based access control using bilinear pairings.
In efforts to protect function privacy, some function-hiding IPFE schemes have been introduced \cite{bishop2015function,abdalla2018multi,shi2023multi}.
All the aforementioned IPFE schemes rely on a trusted central authority for secret key generation.
To enhance privacy and reduce reliance on centralised infrastructure, 
Han et al. \cite{han2023privacy} proposed a decentralised IPFE scheme where function keys are distributed by independently operating authorities without disclosing users' identities.

With the advancement of quantum computing \cite{liu2021efficient,lu2024quantum}, IPFE schemes based on classical hardness assumptions (e.g., factoring problems and discrete logarithm problems) face significant security threats \cite{shor1999polynomial}.
To resist quantum attacks, the first lattice-based IPFE scheme was introduced by Abdalla et al. \cite{abdalla2015simple}. 
In their security model, adversaries are required to submit challenge messages during the initialisation phase.
To achieve adaptive security, Agrawal et al. \cite{agrawal2016fully} introduced the first adaptively secure lattice-based IPFE scheme.
To address challenges in managing public keys and certificates, Abdalla et al. \cite{abdalla2020inner} proposed two lattice-based identity-based IPFE constructions, achieving IND-CPA security in the random oracle and standard models, respectively.
To achieve fine-grained access control, lattice-based IPFE schemes with attribute-based access control were proposed \cite{pal2021attribute,luo2021generic}.
Extending the work of Abdalla et al. \cite{abdalla2020inner}, Cini et al. \cite{cini2024inner} proposed a lattice-based updatable IPFE scheme, where tags are associated with ciphertexts and function keys. 
The function value of the encrypted data can be computed if and only if the tags match, and the scheme further supports the updating of ciphertext tags.
More recently, Lai et al. \cite{lai2025pqhealthcare} introduced an efficient lattice-based IPFE scheme aimed at enabling secure sharing of health data, employing two-stage sampling and rejection sampling techniques to enhance performance and security.

To trace and revoke malicious users who leak their secret keys, Luo et al. \cite{luo2022generic} proposed a generic construction of a trace-and-revoke IPFE scheme that supports public traceability.
However, their initial construction \cite{luo2022generic} can only tolerate bounded collusion, meaning the number of secret keys obtained by a pirate decoder must be limited in advance.
To address this limitation, Luo et al. \cite{luo2024fully} introduced a fully collusion-resistant trace-and-revoke IPFE scheme, building upon \cite{kim2020collusion}.
In the revocable IPFE schemes \cite{luo2022generic,luo2024fully}, direct user revocation is realized by 
specifying a revocation list when encrypting. 
Recently, Zhu et al. \cite{zhu2025revocable} proposed an indirect revocable
hierarchical identity-based IPFE scheme where user revocation and traceability are supported. 

Our scheme is different from the existing revocable IPFE schemes \cite{zhu2025revocable,luo2022generic,luo2024fully}. 
Schemes \cite{zhu2025revocable,luo2022generic,luo2024fully} only support user revocation, 
namely, all function computing rights of a user are revoked, instead of only specific functions.
The IPFE scheme \cite{zhu2025revocable} revokes users by updating the corresponding ciphertexts. 
The scheme \cite{zhu2025revocable} is identity-based and only supports user revocation, but key update is not considered. 
The existing updatable IPFE scheme \cite{luo2024updatable} does not consider revocation, so it cannot support key updates for the specified users.
Moreover, in the schemes \cite{luo2022generic,luo2024fully}, user revocation is direct, namely, the data owner specifies the revocation list when encrypting and 
needs to maintain the revocation lists \cite{attrapadung2009attribute}. This 
is not applicable in a secure data sharing scheme, because the EHR holders cannot 
control the EHR data after outsourcing it to cloud servers \cite{hur2010attribute}. 
The scheme \cite{zhu2025revocable} supports the indirect user revocation, namely, the revocation is conducted by the central authority, but its security relies on the BDHE assumption, which is not post-quantum secure. 
Additionally, the schemes \cite{luo2022generic} and \cite{luo2024fully} are not 
forward secure because revoked users can still compute the function on previously encrypted data. 
Although the scheme \cite{zhu2025revocable} is forward secure, both encryption and ciphertext updating are performed by a CA, which imposes a significant computational burden on the CA. 
 
Table \ref{tab:compare} shows a comparison of features between our scheme and related schemes. 
\begin{table*}[!ht] 
    \caption{Feature Comparison with Related Schemes}
    \label{tab:compare}
    \centering
    \normalsize
    \resizebox{0.9\textwidth}{!}{
    \begin{tabular}{|c|c|c|c|c|c|c|c|c|}
    \hline
    Scheme&Revocation mode&Revocation type&Forward security&Collusion resistance&Key update&Post-quantum security\\
    \hline
    \cite{luo2022generic}&Direct&User revocation&\ding{55}&N/A&\ding{55}&\ding{51}\\
    \hline
    \cite{luo2024fully}&Direct&User revocation&\ding{55}&\ding{51}&\ding{55}&\ding{51}\\
    \hline
    \cite{zhu2025revocable}&Indirect&User revocation&\ding{51}&N/A&\ding{55}&\ding{55}\\
    \hline
    Our scheme&Indirect&Fine-grained revocation&\ding{51}&\ding{51}&\ding{51}&\ding{51}\\
    \hline
    \end{tabular}}
    \end{table*}

\subsection{Organization}
The remainder of this paper is structured as follows. 
Section \ref{section:second} introduces the preliminaries.
In Section \ref{section:third}, we describe the system architecture 
of the flexible EHR sharing scheme.
Section \ref{section:fourth} presents the formal definition and security model of the IPFE-FR scheme.
The detailed construction of our lattice-based IPFE-FR scheme is provided in Section \ref{section:fifth}.
In Section \ref{section:seventh}, we evaluate the scheme through theoretical comparisons and implementation results.
Section \ref{section:sixth} presents the security proof of the proposed scheme.
Finally, Section \ref{section:eighth} concludes the paper.

\section{Preliminaries}\label{section:second}
Table \ref{tab:table_notation} provides a summary of the notations used throughout this paper.

\begin{table*}[!ht]
    \caption{Notation Summary}\label{tab:table_notation}
    \centering
    \begin{tabular}{l|l||l|l}
        \hline
        Notation & Description & Notation & Description  \\
        \hline
        $\lambda$& Security parameter & $\epsilon(\lambda)$ & A negligible function in $\lambda$\\
        $p$ & A prime number & $\top$ & A symbol denoting transposition\\
        $\mathcal{N}$ & The bound of medical institutions & $\bot$ & A symbol denoting termination of algorithm\\
        $ver$ & The version number of system & $\oplus$ & A symbol denoting XOR\\
        $\mathbf{x}$ & A vector & $A(x) \rightarrow y$ & $y$ is obtained by running algorithm $A(\cdot)$ On input $x$.\\
        $\langle \mathbf{x},\mathbf{y} \rangle$ & The inner-product of $\mathbf{x}$ and $\mathbf{y}$ &$PPT$ & Probable polynomial-time\\
        $\mathsf{bin}(x)$ & The binary version of $x$ & $CA$ & Central authority \\
        $\left\lceil a \right\rceil$ & The smallest integer not less than $a$ & $GM$ & Group manager\\
        $\mathsf{id}$ & The user identity & $CS$ & Cloud server\\
        $\mathcal{R}_{\mathbf{x}}$ & The set of revoked users for $\mathbf{x}$ & $EH$ & EHR holder \\
        $\mathsf{pd}$ & A public directory & $MI$ & Medical institution \\
        \hline
    \end{tabular}
    \end{table*}

\subsection{Preliminaries}

$\mathbf{Lattice.}$ 
The lattice $\Lambda$ generated by $n$-linearly independent vectors 
$\mathbf{b}_1,\ldots ,\mathbf{b}_n \in \mathbb{R}^n$ is defined as 
$\Lambda= \{ \sum_{i = 1}^{n} x_i\mathbf{b}_i: x_i \in \mathbb{Z} \}$. 
Moreover, $\Lambda^{\bot} = \{\mathbf{y} \in span(\Lambda): \forall \mathbf{x}\in \Lambda, \left\langle\mathbf{x}, \mathbf{y}\right\rangle \in \mathbb{Z}\}$. 
Given $\mathbf{A} \in \mathbb{Z}_q^{n\times m}$,    
$\Lambda_q^{\bot}(\mathbf{A}) = \{\mathbf{u}\in \mathbb{Z}^m:\mathbf{Au} = \mathbf{0}\ (\mathrm{mod\ } q)\}$ 
and 
$\Lambda_q^{\mathbf{z}}(\mathbf{A}) = \{\mathbf{u}\in \mathbb{Z}^m:\mathbf{Au} = \mathbf{z}\ (\mathrm{mod\ } q)\}$, 
where $q$ is a prime.

$\mathbf{Matrix\ Norms.}$ Let $\left\lVert \mathbf{u} \right\rVert$ denote the $l_2$ norm of a vector $\mathbf{u}$.
Given $\mathbf{A} \in \mathbb{Z}^{k \times m}$, 
$\widetilde{\mathbf{A}}$ represents the Gram-Schmidt orthogonalization of its columns.
Moreover, 
$\left\lVert \mathbf{A} \right\rVert$, $\left\lVert \mathbf{A} \right\rVert_{2}$ 
and $s_1(\mathbf{A})$ denotes the longest column, the operator norm and the spectral norm of $\mathbf{A}$, 
respectively. 

$\mathbf{Gaussian\ Distribution.}$
For any lattice $\Lambda \in \mathbb{Z}^n$, 
the discrete Gaussian distribution over $\Lambda$ 
is defined as 
$\mathcal{D}_{\Lambda,\sigma,\mathbf{c}}(\mathbf{y}) = \rho_{\sigma,\mathbf{c}}(\mathbf{y})/\rho_{\sigma,\mathbf{c}}(\Lambda)$, where $\mathbf{y} \in \Lambda$, 
$\mathbf{c}\in \mathbb{R}^n$, $\sigma>0$, 
$\rho_{\sigma,\mathbf{c}}(\mathbf{y}) = \mathrm{exp}(-\pi \left\lVert \mathbf{y}-\mathbf{c} \right\rVert^2/\sigma^2 )$, 
and $\rho_{\sigma,\mathbf{c}}(\Lambda) = \sum_{\mathbf{x}\in \Lambda}\rho_{\sigma,\mathbf{c}}(\mathbf{x})$.

\begin{lemma}[Preimage Sampleable Functions \cite{gentry2008trapdoors}] \label{TrapGen}
    Given prime $q=poly(n)$, $m \geq O (n \mathrm{log\ }q)$, and $s \geq \lVert \widetilde{\mathbf{T}_\mathbf{A}} \rVert \cdot \omega(\sqrt{\mathrm{log\ } m} )$, 
    there are two PPT algorithms 
    $\mathsf{TrapGen}$ and $\mathsf{SamplePre}$. 
    
$\mathsf{TrapGen}(1^n,1^m,q) \to (\mathbf{A},\mathbf{T}_\mathbf{A})$, 
     where $\mathbf{A} \in \mathbb{Z}^{n\times m}_q$ is statistically close to uniform, 
     and $\mathbf{T}_\mathbf{A} \subset \Lambda ^{\bot}_q(\mathbf{A})$ is a basis satisfying $\lVert \widetilde{\mathbf{T}_\mathbf{A}} \rVert \leq  O(\sqrt{n \mathrm{log\ }q} )$.
     
$\mathsf{SamplePre}(\mathbf{A},\mathbf{T}_\mathbf{A},\mathbf{U},s) \to \mathbf{Z}$, 
    where $\mathbf{Z}\in\mathbb{Z}^{m\times l}$ and $\mathbf{U}\in \mathbb{Z}_q^{n \times l}$ such that $\mathbf{U}=\mathbf{A}\mathbf{Z}$.

    Furthermore, the distributions $D_1$ and $D_2$ are statistically indistinguishable:\\
\begin{equation*}
    \begin{aligned}
    &D_1=(\mathbf{A},\mathbf{Z'},\mathbf{A}\mathbf{Z}'), \mathrm{where\ } \mathbf{A} \gets \mathbb{Z}^{n\times m}_q, 
    \mathbf{Z}' \gets \mathcal{D}_{\mathbb{Z}^{m \times l},s};\\
    &D_2=(\mathbf{A},\mathbf{Z},\mathbf{U}), \mathrm{where\ } (\mathbf{A},\mathbf{T}_\mathbf{A}) \gets \mathsf{TrapGen}(1^n,1^m,q), \\
    & \hspace{1.8cm} \mathbf{Z} \gets \mathsf{SamplePre}(\mathbf{A},\mathbf{T}_\mathbf{A},\mathbf{U},s), \mathbf{U} \gets \mathbb{Z}_q^{n \times l}.
    \end{aligned}
\end{equation*}  
\end{lemma}

\begin{lemma}[Bounding Gaussian Noise \cite{micciancio2007worst}] \label{Bound Gaussian}
    For any n-dimension lattice $\varLambda$, $\mathbf{c}\in span(\Lambda)$, real $\epsilon \in (0,1)$ 
    and $s \geq \eta(\Lambda)$: 
    $$
        \underset{\mathbf{x}{\gets}\mathcal{D}_{\Lambda,s,\mathbf{c}}}{\mathrm{Pr}}
        [\lVert \mathbf{x}-\mathbf{c} \rVert > s\sqrt{n}] \leq \frac{1+\epsilon}{1-\epsilon} \frac{1}{2^n}.
    $$
    Furthermore, for any $\omega(\sqrt{\mathrm{log\ }n})$, there exists a negligible function $\epsilon(n)$ satisfying 
    $\eta_{\epsilon}(\mathbb{Z}) \leq \omega(\sqrt{\mathrm{log\ }n})$. Specifically, 
    when sampling integers from the discrete Gaussian distribution $\mathcal{D}_{\mathbb{Z},s,c}$, 
    we have that: 
    $$
        \underset{x{\gets}\mathcal{D}_{\mathbb{Z},s,c}}{\mathrm{Pr}}[\lvert x-c \rvert > s\cdot t]
        \leq \mathsf{negl}(n).
    $$
    where $\epsilon \in(0,\frac{1}{2})$, $s\geq \eta_{\epsilon}(\mathbb{Z})$ 
    and $t\geq \omega(\sqrt{\mathrm{log\ }n})$. 
\end{lemma}

\begin{lemma}[Bounding Spectral Norm of a Gaussian Matrix \cite{ducas2014improved}] \label{BoundGauss}
    Let $\mathbf{Z} \in \mathbb{R}^{n\times m}$ be a sub-Gaussian random matrix with parameter $\rho$. 
    Then, there exists a universal constant $C\approx \frac{1}{\sqrt{2\pi}}$ such that for any $t \geq 0$, we have 
    $s_1(\mathbf{Z}) \leq C \cdot \rho (\sqrt{n}+\sqrt{m}+t)$ except with probability at most $\frac{2}{e^{\pi t^2}}$. 
\end{lemma}

$\mathbf{Vector\ Decomposition.}$ For integers $n, m, p \geq 1$, any vector can be decomposed into its bitwise representation and formally expanded as a linear combination of powers of two. 
Specifically, the details are as follows: 

$\mathsf{BitD}_{p}(\mathbf{x})$: 
    On input $\mathbf{x}\in \mathbb{Z}_{p}^{n}$, 
    output $(\mathbf{x}_{0},\ldots,\mathbf{x}_{\left\lceil \mathrm{log\ } p \right\rceil-1}) \in 
    \{ 0,1\}^{n\left\lceil \mathrm{log\ } p \right\rceil}$ 
    where $\mathbf{x}_{i}\in \{0,1 \}^n$ and 
    $\mathbf{x}= \sum_{i=0}^{\left\lceil \mathrm{log\ } p \right\rceil-1} 2^i \cdot \mathbf{x}_{i}$ (mod $p$).

$\mathsf{PowerT}_{p}(\mathbf{y})$: 
     On input $\mathbf{y}\in \mathbb{Z}^{n}_p$, 
    output $(\mathbf{y},2\cdot \mathbf{y},\ldots, 2^{\left\lceil \mathrm{log\ }p -1\right\rceil }\cdot \mathbf{y}) 
    \in \mathbb{Z}^{n \left\lceil \mathrm{log\ } p\right\rceil}_p$. 

\begin{lemma} \label{Decomposition}
    For any $\mathbf{x},\mathbf{y} \in \mathbb{Z}^{n}_p$, it holds that $\langle \mathbf{x},\mathbf{y} \rangle
    =\langle \mathsf{BitD}_{p}(\mathbf{x}), \mathsf{PowerT}_{p}(\mathbf{y})\rangle$ (mod $p$). 
\end{lemma}
Lemma \ref{Decomposition} also extends to matrices, which can be viewed as collections of column or row vectors.  

$\mathbf{Learning\ with\ Errors\ (LWE)}$ \cite{regev2009lattices}. For integers $n,m$, a prime $q$, 
and $\mathcal{D}_{\mathbb{Z}^m,\alpha q}$ with $\alpha\in \left( 0,1 \right)$, 
the $\mathrm{LWE}_{n,q,m,\alpha}$ problem is to distinguish the distribution 
$(\mathbf{A},\mathbf{A}^{\top}\cdot \mathbf{s}+\mathbf{e})$ from $(\mathbf{A},\mathbf{u})$ 
where $\mathbf{A} \gets \mathbb{Z}_{q}^{n \times m}$, $\mathbf{s}\gets \mathbb{Z}_{q}^{n}$, 
$\mathbf{e}\gets \mathcal{D}_{\mathbb{Z}^{m},\alpha q}$ and $\mathbf{u} \gets \mathbb{Z}_{q}^m$. 
LWE assumption holds if the advantage of all PPT $\mathcal{A}$ in solving the above problem 
is negligible, namely
\begin{equation*}
Adv^{LWE}_{\mathcal{A}}(\lambda) = 
\left\lvert 
\begin{aligned}
&\mathrm{Pr}\left[ \mathcal{A}(\mathbf{A},\mathbf{A}^{\top}\cdot \mathbf{s}+\mathbf{e})=1\right] \\
&\ \ \ \ -\mathrm{Pr}\left[ \mathcal{A}(\mathbf{A},\mathbf{u})=1\right]
\end{aligned}
\right\rvert
\leq \epsilon(\lambda).
\end{equation*} 

\subsection{ALS-IPFE} \label{ALS} 
Our IPFE-FR scheme builds upon the ALS-IPFE scheme introduced in \cite{agrawal2016fully}. 
To optimize the parameters of the ALS-IPFE scheme and simplify its proof, 
Wang et al. \cite{wang2019fe} applied a re-randomization technique \cite{katsumata2016partitioning} 
to modify the distribution of the master secret key. 
For the sake of completeness, we briefly revisit the modified ALS-IPFE scheme \cite{abdalla2020inner} in the following.
\begin{itemize}
    \item[$\bullet$] $\mathbf{Setup}(1^l,1^\lambda,p):$ On input $1^l$, $1^\lambda$, 
    set integers $n$, $m$, $X,Y$, $K=lXY$, prime $p\geq 2$, reals $\alpha\in(0,1)$, $\rho>0$, 
    spaces $\mathcal{X} = \{0,\ldots X-1\}^l$ and 
    $\mathcal{Y} = \{0,\ldots Y-1\}^l$. 
    Sample  
    $\mathbf{A} \gets \mathbb{Z}_p^{n\times m}$, $\mathbf{Z}\gets \mathcal{D}_{\mathbb{Z}^{m\times l},\rho}$, and 
    compute $\mathbf{U}= \mathbf{A}\mathbf{Z} \in \mathbb{Z}_p^{n\times l}$. 
    Let the master secret key $\mathsf{msk}=\mathbf{Z}$ and the master public key 
    $\mathsf{mpk}=(\mathbf{A},\mathbf{U})$. 

    \item[$\bullet$] $\mathbf{KeyGen}(\mathsf{msk},\mathbf{x}):$ On input  $\mathsf{msk}$ and a vector $\mathbf{x} \in \mathcal{X}$, 
    compute  
    $\mathbf{z}_{\mathbf{x}} = \mathbf{Z} \cdot \mathbf{x} \in \mathbb{Z}^m$, and
    return $\mathsf{sk}_{\mathbf{x}}=(\mathbf{x},\mathbf{z}_{\mathbf{x}})$. 
    
    \item[$\bullet$] $\mathbf{Enc}(\mathsf{mpk},\mathbf{y}):$ 
    On input $\mathsf{mpk}$ and a vector $\mathbf{y}\in \mathcal{Y}$, 
    sample $\mathbf{s}\gets \mathbb{Z}_p^n$, $\mathbf{e}_1 \gets \mathcal{D}_{\mathbb{Z}^m,\sigma}$,  
    $\mathbf{e}_2 \gets \mathcal{D}_{\mathbb{Z}^l,\sigma}$, compute 
    $\mathbf{c}_1 = \mathbf{A}^{\top}\cdot \mathbf{s}+\mathbf{e}_1 \in \mathbb{Z}_p^m$, 
    $\mathbf{c}_2 = \mathbf{U}^{\top} \cdot \mathbf{s} +\mathbf{e}_2+ \left\lfloor \frac{p}{K} \right\rfloor \cdot \mathbf{y} \in \mathbb{Z}_p^l$, and return the ciphertext $\mathbf{ct}=(\mathbf{c}_1,\mathbf{c}_2)$. 

    \item[$\bullet$] $\mathbf{Dec}(\mathsf{mpk},\mathsf{sk}_{\mathbf{x}},\mathbf{ct}):$ 
    On input the $\mathsf{mpk}$, $\mathsf{sk}_{\mathbf{x}}$ and $\mathbf{ct}$, 
    compute $\mu' = {\mathbf{x}}^{\top} \cdot \mathbf{c}_{2} -  
    \mathbf{z}^{\top}_{\mathbf{x}} \cdot \mathbf{c}_{1} $ mod $p$, and 
    output $\mu \in\{0,\ldots,K-1 \}$, where $\mu$ is the value 
    that minimizes $\lvert \left\lfloor \frac{p}{K} \right\rfloor  \cdot \mu - \mu' \rvert$.
\end{itemize}
Following the \cite{abdalla2020inner} (Appendix A), 
the parameters can be set as: 
$\frac{p}{2K} \geq l \sqrt{l} X \omega(\mathrm{log^2\ }n),
\rho \geq \omega(\sqrt{\mathrm{log\ }n}),
\sigma \geq 2 C\alpha p (\sqrt{n}+\sqrt{m}+\sqrt{l}), 
m=2n\cdot \mathrm{log\ }p$, 
where $C$ denoting the constant specified in the Lemma \ref{BoundGauss}. 

\subsection{Pseudorandom Function}
Let $\mathcal{K}$, $\mathcal{P}$ and $\mathcal{V}$  be a key-space, a domain and a range, respectively.
Given the security parameter $\lambda$, a pseudorandom function (PRF) is an efficiently computable 
function $F: \mathcal{K} \times \mathcal{P} \to \mathcal{V}$ satisfying the following security property:  
for all PPT adversaries $\mathcal{A}$, the advantage 
\begin{equation*}
    Adv_{\mathcal{A}}(\lambda) = 
    \left\lvert
    \begin{aligned}
    \mathrm{Pr}[\mathcal{A}^{F(k,\cdot)}(1^\lambda)=1]\\
    - \mathrm{Pr}[\mathcal{A}^{R(\cdot)}(1^\lambda)=1] 
    \end{aligned}
    \right\lvert
    \leq \epsilon(\lambda),
\end{equation*}
is negligible, 
where $k \gets \mathcal{K}$ and $R: \mathcal{P} \to \mathcal{V}$ is a random function.

\section{The System Model of  flexible EHR
sharing scheme}\label{section:third}
Fig.\ref{figure:overview} illustrates the system model of flexible EHR sharing scheme. 
\begin{figure*}[!ht]
    \centering
    \includegraphics[width=0.7\linewidth]{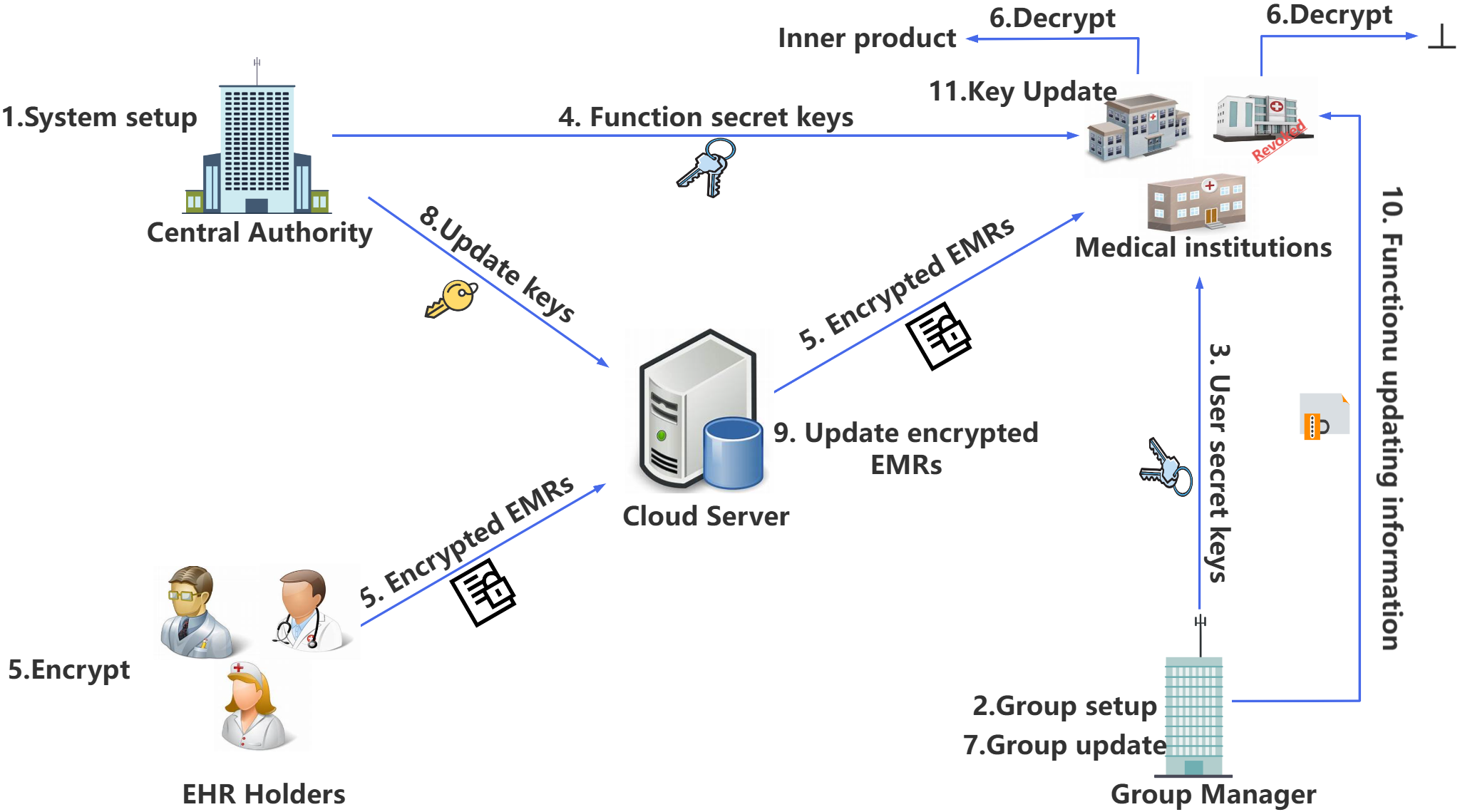}
    \caption{System model of flexible EHR sharing scheme}\label{figure:overview}
\end{figure*}

The system model consists of five entities: 
a central authority, a group manager, a cloud server, EHR holders and medical institutions. 

Central authority ($CA$): $CA$ is 
    a trusted entity that sets up the system and generates function keys for medical institutions.
    Additionally, it generates update keys for the cloud server when revocation is required.
    
Group manager ($GM$): $GM$ is a trusted entity that generates user keys for medical institutions 
    and  broadcasts 
    function update information to medical institutions. 
    
Cloud server ($CS$): $CS$ is a semi-honest entity that stores encrypted EHR data. 
    When revocation occurs, $CS$ obtains the update key from the $GM$ and updates the encrypted EHR data. 
    
EHR holders ($EH$s): $EH$s are entities with EHR  
    that usually contain sensitive personal information, such as age, gene data, etc. 
    To share EHR data securely, $EH$s
    encrypt them and upload encrypted EHR data to $CS$ honestly.  
    
Medical institutions ($MI$s): $MI$s, such as doctors, researchers and nurses, are entities that want to use these EHR data.  
    To access EHR data, $MI$s request function keys and user keys from $CA$ and 
    $GM$, respectively. 
    Then, unrevoked $MI$s can compute the inner product of the vector 
    contained in the encrypted EHR data and the vector associated with their secret key. 

The workflow of our scheme is as follows. 
    $CA$ initialises the system and 
    generates a master secret-public key pair and public parameters. 
    Then, $GM$  
    generates the group secret-public key pair.
    To obtain the user keys, $MI$s send their identities 
    to $GM$.  
    $GM$ issues user keys to $MI$s.  
    To obtain function keys, $MI$s send their identities and 
    vectors to $CA$.  
    $CA$ generates function keys for $MI$s. 
    $EH$s encrypt EHR data and upload the 
    encrypted EHR data to $CS$. 
    $MI$s can obtain encrypted EHR data from $CS$. 
    If a $MI$'s function computation rights are not revoked, 
    it can calculate the inner product of the EHR data and the vector associated with function key.  
    When revocation is required, $GM$ updates the system version and generates new version group secret-public key pair. 
    $CA$ generates the update key for $CS$, and
    $CS$ updates the encrypted EHR data.  
    Then, $GM$ selects a function to be updated, sets the revocation list of the function, generates function update information, and 
    broadcasts it to $MI$s. 
    Using the function update information, unrevoked $MI$s can update their function keys to the new version. 

\section{Formal Definition and Security Model}\label{section:fourth}
This section introduces the formal definition and security model of IPFE-FR. 

\subsection{Formal Definition}
IPFE-FR is composed of the following algorithms:

$\mathbf{SystemSetup}(1^\lambda,\mathcal{N}) \to (\mathsf{msk},\mathsf{mpk},\mathsf{pp})$. 
    $CA$ 
    takes 
    a security parameter $\lambda$ and a bound $\mathcal{N}$ of the number of users as inputs, and outputs the master secret-public key 
    pair $(\mathsf{msk},\mathsf{mpk})$ and public parameters $\mathsf{pp}$. 
    For brevity, we omit the public parameters $\mathsf{pp}$ from the remaining algorithms' inputs.
 
$\mathbf{GroupSetup}(\mathsf{mpk}) \to (\mathsf{gsk}_1,\mathsf{gpk}_1)$. 
    $GM$  
    takes the master public key $\mathsf{mpk}$ as inputs, and outputs
    the group secret-public key pair $(\mathsf{gsk}_1,\mathsf{gpk}_1)$, where $1$ denotes the initial system version number. 
    When revocation occurs, the group secret-public key pair will be updated to a new version which increases by $1$. 
    We denote the current system version number by $ver$ in our scheme.
    Notably, $\mathsf{gsk}_1$ consists of two parts: $\mathsf{guk}$ and $\mathsf{gfk_1}$. 
    $\mathsf{guk}$ is constant and is used to generate user keys for $MI$s, 
    while $\mathsf{gfk}_1$ is changed with the system version number and is used to generate function update information for $MI$s.

$\mathbf{UKeyGen}(\mathsf{guk},\mathsf{id}) \to \mathsf{usk}_{\mathsf{id}}$. 
    $GM$ takes the part group secret key $\mathsf{guk}$ and 
    a $MI$'s identity $\mathsf{id}$ as inputs, 
    and generates the user key $\mathsf{usk}_{\mathsf{id}}$ 
    for the $MI$. 

$\mathbf{FKeyGen}(\mathsf{msk}, \mathsf{mpk},\mathsf{gpk}_{ver},\mathbf{x},\mathsf{id}) \to \mathsf{fsk}_{\mathbf{x},\mathsf{id},ver}$. 
    $CA$ takes the master secret-public key pair $(\mathsf{msk},\mathsf{mpk})$,  
    the group secret key $\mathsf{gpk}_{ver}$, 
    a vector $\mathbf{x}$ and 
    a $MI$'s identity $\mathsf{id}$ as inputs, and generates the function key $\mathsf{fsk}_{\mathbf{x},\mathsf{id},ver}$ 
    for the $MI$. 

$\mathbf{Enc}(\mathsf{mpk},\mathsf{gpk}_{ver},\mathbf{y})\to \mathsf{CT}_{ver}$. 
    $EH$ 
    takes the master public key $\mathsf{mpk}$, 
    the group public key $\mathsf{gpk}_{ver}$ and 
    EHR data $\mathbf{y}$ as inputs,  
    and uploads encrypted EHR data $\mathsf{CT}_{ver}$ 
    to $CS$. 

$\mathbf{Dec}(\mathsf{CT}_{ver},\mathsf{usk}_{\mathsf{id}},\mathsf{fsk}_{\mathbf{x},\mathsf{id},ver}) 
    \to \langle  \mathbf{x},\mathbf{y} \rangle$. 
    $MI$ 
    takes the encrypted EHR data $\mathsf{CT}_{ver}$, 
    a user key $\mathsf{usk}_{\mathsf{id}}$ and 
    the function key $\mathsf{fsk}_{\mathbf{x},\mathsf{id},ver}$ as inputs, 
    and outputs the inner product $\langle \mathbf{x},\mathbf{y} \rangle$. 

$\mathbf{GroupUpdate}
    (\mathsf{mpk},\mathsf{gsk}_{ver},\mathsf{gpk}_{ver})\to (\mathsf{gsk}_{ver+1},$ $\mathsf{gpk}_{ver+1})$. 
    $GM$ takes the master public key $\mathsf{mpk}$ and 
    the group secret-public key pair $(\mathsf{gsk}_{ver},\mathsf{gpk}_{ver})$ as inputs, 
    and generates the 
    new version group secret-public key pair $(\mathsf{gsk}_{ver+1},\mathsf{gpk}_{ver+1})$ 
    for $MI$s. 

$\mathbf{UptKeyGen}(\mathsf{msk}, \mathsf{mpk},\mathsf{gpk}_{ver+1}) 
    \to \mathsf{uptk}_{ver+1}$. 
    $CA$ takes the master secret-public key pair $(\mathsf{msk},\mathsf{mpk})$ and the group public key 
    $\mathsf{gpk}_{ver+1}$ as inputs, and generates the update key $\mathsf{uptk}_{ver+1}$ for $CS$. 

$\mathbf{CTUpdate}(\mathsf{uptk}_{ver+1},\mathsf{CT}_{ver})\to \mathsf{CT}_{ver+1}$. 
    $CS$ takes the update key $\mathsf{uptk}_{ver+1}$ and the encrypted EHR data $\mathsf{CT}_{ver}$ as inputs,  
    and outputs the updated EHR data $\mathsf{CT}_{ver+1}$. 

$\mathbf{FUpdate}(\mathsf{mpk},\mathsf{gfk}_{ver},\mathsf{gfk}_{ver+1},\mathbf{x},\mathcal{R}_{\mathbf{x}}) 
    \to \mathsf{UPI}_{\mathbf{x},{ver+1}}$.
    $GM$ takes the master public key pair $\mathsf{mpk}$, 
    the part group secret keys $\mathsf{gfk}_{ver}, \mathsf{gfk}_{ver+1}$ and 
    the vector $\mathbf{x}$, the set of revoked users $\mathcal{R}_{\mathbf{x}}$ for $\mathbf{x}$ as inputs, 
    and outputs the function update information $\mathsf{UPI}_{\mathbf{x},ver+1}$ for $MI$s. 

$\mathbf{KeyUpdate}(\mathsf{usk}_{\mathsf{id}},\mathsf{fsk}_{\mathbf{x},\mathsf{id},ver}, \mathsf{UPI}_{\mathbf{x},{ver+1}}) 
    \to \mathsf{fsk}_{\mathbf{x},\mathsf{id},{ver+1}}$. 
     $MI$ takes 
    the user key $\mathsf{usk}_{\mathsf{id}}$, the function key $\mathsf{fsk}_{\mathbf{x},\mathsf{id},ver}$ and 
    the function update information $\mathsf{UPI}_{\mathbf{x},ver+1}$ as inputs, 
    and outputs the function key $\mathsf{fsk}_{\mathbf{x},\mathsf{id},ver+1}$. 

\begin{definition} An inner-product functional encryption with fine-grained revocation (IPFE-FR) is correct if

    \begin{equation*}
    \begin{split}
        \mathsf{Pr}\left[\begin{array}{l|l}
        &\mathbf{SystemSetup}(1^\lambda,\mathcal{N}) \to(\mathsf{msk}, \\
        &\mathsf{mpk},\mathsf{pp});\\
        &\mathbf{GroupSetup}(\mathsf{mpk}) \to(\mathsf{gsk}_1,\\
        &\mathsf{gpk}_1);\\
        &\mathbf{UKeyGen}(\mathsf{guk},\mathsf{id}) \to \mathsf{usk}_{\mathsf{id}};\\
        &\mathbf{FKeyGen}(\mathsf{msk}, \mathsf{mpk},\mathsf{gpk}_{ver},\mathbf{x},\\
        \mathbf{Dec}(&\mathsf{id}) \to \mathsf{fsk}_{\mathbf{x},\mathsf{id},ver};\\
        \mathsf{CT}_{ver},&\mathbf{Enc}(\mathsf{mpk},\mathsf{gpk}_{ver},\mathbf{y})\to\mathsf{CT}_{ver}; \\
        \mathsf{usk}_{\mathbf{x},\mathsf{id}}, &\mathbf{GroupUpdate} 
        (\mathsf{mpk},\mathsf{gsk}_{ver},\\
        \mathsf{fsk}_{\mathbf{x},\mathsf{id},ver},&\mathsf{gpk}_{ver}) \to (\mathsf{gsk}_{ver+1},\mathsf{gpk}_{ver+1});\\
        ) \to&\mathbf{UptKeyGen}(\mathsf{msk}, \mathsf{mpk},\mathsf{gpk}_{ver+1})\\
        \langle  \mathbf{x},\mathbf{y} \rangle&\to\mathsf{uptk}_{ver+1}\\
        &\mathbf{CTUpdate}(\mathsf{uptk}_{ver+1},\mathsf{CT}_{ver})\to\\
        &\mathsf{CT}_{ver+1};\\
        &\mathbf{FUpdate}(\mathsf{mpk},\mathsf{gfk}_{ver},\mathsf{gfk}_{ver+1},\\
        &\mathbf{x},\mathcal{R}_{\mathbf{x}})\to \mathsf{UPI}_{\mathbf{x},{ver+1}};\\
        &\mathbf{KeyUpdate}(\mathsf{usk}_{\mathbf{x},\mathsf{id}},\mathsf{fsk}_{\mathbf{x},\mathsf{id},ver},\\ &\mathsf{UPI}_{\mathbf{x},{ver+1}})\to\mathsf{fsk}_{\mathbf{x},\mathsf{id},{ver+1}};\\
        \end{array}
        \right]=1
    \end{split}
    \end{equation*}

\end{definition}

\subsection{Security Model} \label{model}
The sIND-CPA security game between the adversary $\mathcal{A}$ and the challenger $\mathcal{C}$ 
is defined as follows. 

$\mathbf{Init.}$ $\mathcal{A}$ submits a version number $ver^*$, a vector $\mathbf{x}^*$ 
and a set of revoked users $\mathcal{R}^*$ to  $\mathcal{C}$. 

$\mathbf{Setup.}$ $\mathcal{C}$ runs $\mathbf{SystemSetup}(1^\lambda,1^\mathcal{N}) \to (\mathsf{msk},\mathsf{mpk},\mathsf{pp})$ and 
$\mathbf{GroupSetup}(\mathsf{mpk})
\to (\mathsf{gsk}_1,\mathsf{gpk}_1)$. 
Then, $\mathcal{C}$ runs $\mathbf{GroupUpdate}(\mathsf{mpk},\mathsf{gsk}_{ver},\mathsf{gpk}_{ver}) \to (\mathsf{gsk}_{ver+1},
\mathsf{gpk}_{ver+1})$ for $ver = 1,\dots,ver^*-1$. 
Finally, $\mathcal{C}$ keeps $(\mathsf{msk}, \{\mathsf{gsk}_{ver}\}_{ver\in \left[1,ver^*\right]})$ and 
sends $(\mathsf{mpk},\mathsf{pp},\{\mathsf{gpk}_{ver}\}_{ver\in \left[1,ver^*\right]})$ to $\mathcal{A}$. 

$\mathbf{Phase\ 1.}$ 

$\mathsf{User\ Key\ Query}.$ $\mathcal{A}$ adaptively submits an identity $\mathsf{id}$ 
to $\mathcal{C}$.  
Then, $\mathcal{C}$ runs $\mathbf{UKeyGen}(\mathsf{guk},\mathsf{id}) \to \mathsf{usk}_{\mathsf{id}}$ and 
forwards $\mathsf{usk}_{\mathsf{id}}$ to $\mathcal{A}$. 

$\mathsf{Function\ Key\ Query}.$ $\mathcal{A}$ adaptively submits $(\mathbf{x},\mathsf{id},ver)$ 
to $\mathcal{C}$, where $\mathbf{x}\in \mathbb{Z}_p^l$ and $ver \in \left[ 1,ver^*\right]$, 
with the restriction that $(\mathbf{x},ver) \neq (\mathbf{x}^*,ver^*)$. 
Then, $\mathcal{C}$ runs 
$\mathbf{FKeyGen}(\mathsf{msk}, \mathsf{mpk},\mathsf{gpk}_{ver},\mathbf{x},\mathsf{id}) \to \mathsf{fsk}_{\mathbf{x},\mathsf{id},ver}$
and 
forwards $\mathsf{fsk}_{\mathbf{x},\mathsf{id},ver}$ to $\mathcal{A}$. 
$\mathcal{C}$ 
adds $\mathbf{x}$ into a (initially empty) set $T_x$ if $\mathbf{x} \notin T_x$. 

$\mathsf{Update\ Key\ Query}.$ $\mathcal{A}$ adaptively submits a version number $ver \in [1,ver^{*}-1]$ to $\mathcal{C}$. 
Then, $\mathcal{C}$ runs $\mathbf{UptKeyGen}(\mathsf{msk}, \mathsf{mpk},\mathsf{gpk}_{ver+1}) 
\to \mathsf{uptk}_{ver+1}$, and returns the update key $\mathsf{uptk}_{ver+1}$ to $\mathcal{A}$. 

$\mathsf{Function\ Update\ Query.}$ $\mathcal{A}$ submits a version number $ver \in [1,ver^{*}-1]$, 
a vector $\mathbf{x} \in T_{x}$ 
and a set of revoked users $\mathcal{R}_{\mathbf{x}}$ for $\mathbf{x}$, 
with the restriction that $\mathcal{R}^* \subseteq \mathcal{R}_{\mathbf{x}}$ 
if $(\mathbf{x},ver) = (\mathbf{x}^*,ver^*-1)$. 
Then, $\mathcal{C}$ runs $\mathbf{FUpdate}(\mathsf{mpk},\mathsf{gfk}_{ver},\mathsf{gfk}_{ver+1},\mathbf{x},\mathcal{R}_{\mathbf{x}}) \to \mathsf{UPI}_{\mathbf{x},ver+1}$, 
and forwards $\mathsf{UPI}_{\mathbf{x},ver+1}$ to $\mathcal{A}$. 

$\mathbf{Challenge.}$ $\mathcal{A}$ submits two vectors $\mathbf{y}_{0}$ and $\mathbf{y}_{1}$,   
constrained such that for all $\mathbf{x} \in T_x \setminus \{\mathbf{x}^*\}$, the equality $\langle \mathbf{x}, \mathbf{y}_{0} \rangle = \langle \mathbf{x}, \mathbf{y}_{1} \rangle$ holds. 
$\mathcal{C}$ samples a random bit $b$ from $\{0,1\}$, runs $\mathbf{Enc}(\mathsf{mpk},\mathsf{gpk}_{ver^*},\mathbf{y}_b)\to \mathsf{CT}_{ver^*}$, and 
returns $\mathsf{CT}_{ver^*}$ to $\mathcal{A}$.

$\mathbf{Phase\ 2.}$ $\mathcal{A}$ is allowed to make queries  
    for user keys, function keys, update keys, and function update information, all of which are answered by
    $\mathcal{C}$ following the procedure defined in $\mathbf{Phase\ 1}$. 

$\mathbf{Guess.}$ $\mathcal{A}$ produces a guess $b'\in\{0,1\}$ for the bit $b$. 
The adversary $\mathcal{A}$ wins If $b=b'$. 
       
    

Within this security model, to ensure the forward security of IPFE-FR, the adversary $\mathcal{A}$ 
is permitted to adaptively query function keys $\mathsf{fsk}_{\mathbf{x}^*,\mathsf{id},ver}$, 
with the restriction that $ver\in [1,ver^{*}-1]$. 
\begin{definition}
    An inner-product functional encryption with fine-grained revocation scheme is sIND-CPA secure 
    if the advantage of all PPT adversaries $\mathcal{A}$ 
    in winning the above game is negligible, namely:
    \begin{align*}
        Adv^{sIND-CPA}_{\mathcal{A}}(\lambda) = \left|Pr[b=b']-\frac{1}{2} \right| \leq \epsilon(\lambda).
    \end{align*}
\end{definition}

\section{Concrete Construction}\label{section:fifth}
To instantiate the system framework of flexible EHR sharing scheme
and resist quantum attacks, we propose a lattice-based instantiation of IPFE-FR. 
We first present a high-level overview of the lattice-based IPFE-FR, 
and then show the concrete construction of it in Fig. \ref{figure:CC}.

$\mathbf{High-level\ overview.}$ 
$CA$ runs $\mathbf{SystemSetup}$ to initialize the flexible EHR sharing scheme. 
Suppose that $H_1: \{0,1\}^* \to \mathbb{Z}_p^{l_2}$, $H_2: \mathbb{Z}_p \to \{ 0,1\} ^ t $ are cryptographic hash functions, 
and 
$\mathsf{PRF}:\mathbb{Z}_p \times \mathcal{X} \to \mathbb{Z}_p^{l_2}$ is a pseudorandom function. 
In $\mathbf{SystemSetup}$, $CA$ runs $\mathsf{TrapGen}(1^n,1^m,p)$ to obtain 
a trapdoor $\mathbf{T}_\mathbf{A}$, which is associated with $\mathbf{A}\in \mathbb{Z}_q^{n\times m}$. 
Moreover, $CA$ selects $k_{p} \gets \mathbb{Z}_p$ as the key of $\mathsf{PRF}$. 
Then, $CA$ sets master secret key $\mathsf{msk}=(\mathbf{T}_{\mathbf{A}},k_p)$ 
and master public key $\mathsf{mpk} = (\mathbf{A},\mathbf{V},\mathbf{C})$. 

Then, $GM$ sets group secret key $\mathsf{gsk}_1 = (\mathbf{D},\mathbf{B}_1)$ and 
$\mathsf{gpk}_1 = (\mathbf{F} ,\mathbf{U}_1)$, where $1$ denotes the initial system version number.  
Notably, $\mathsf{gsk}_1$ consists of two parts: $\mathsf{guk}=\mathbf{D}$ and $\mathsf{gfk}_1 = \mathbf{B}_1$. 
$\mathsf{guk}$ is constant and is used to generate user keys for $MI$s, 
while $\mathsf{gfk}_1$ is changed with the system version number and is used to generate function update information for $MI$s.

To obtain the user key from $GM$, $MI$ sends his/her identity $\mathsf{id}$ 
to $GM$. 
$GM$ uses $\mathsf{guk}$ to compute $\mathbf{u}_{\mathsf{id}} = \mathbf{D} \cdot \mathbf{id}$ and 
generate the user key $\mathsf{usk}_{\mathsf{id}}= 
(\mathbf{id}, \mathbf{u}_{\mathsf{id}})$ for $MI$, where $\mathbf{id} = H_1(\mathsf{id})$. 

To obtain the function key, $MI$ sends his/her identity $\mathsf{id}$ 
and a vector $\mathbf{x}\in \mathcal{X}$ to $CA$. 
Using $\mathbf{T}_{\mathbf{A}}$, $CA$ can sample $\mathbf{Z}_{ver} \in \mathcal{D}_{\mathbb{Z}^{m\times l_1},\rho_1}$ such 
that $\mathbf{C}+\mathbf{U}_{ver} = \mathbf{A}\cdot \mathbf{Z}_{ver}$, and further compute 
$\mathbf{f}_{\mathbf{x},ver} = \mathbf{Z}_{ver} \cdot \mathbf{x} \in \mathbb{Z}^{m} $. 
Moreover, to resist user collusion attack, $CA$ runs $\mathsf{PRF}(k_p,\mathbf{x})$ to generate a secret
vector $\mathbf{t}_{\mathbf{x}} \in \mathbb{Z}_p^{l_2}$, and encrypts it as 
$(pd_{\mathbf{x},1},pd_{\mathbf{x},2})$, 
where $pd_{\mathbf{x},1} = \mathbf{V}^{\top} \cdot \mathbf{s}_1 + \mathbf{e}_{1} \in \mathbb{Z}_{q}^{m}$, 
$pd_{\mathbf{x},2} = \mathbf{F}^{\top} \cdot \mathbf{s}_1 + \mathbf{e}_{2} + p^{k-1} \cdot \mathbf{t}_{\mathbf{x}} \in \mathbb{Z}_{q}^{l_2}$. 
Using $\mathsf{usk}_{\mathsf{id}}$, a $MI$ can only computes $\langle\mathbf{id}, \mathbf{t}_{\mathbf{x}}\rangle$ mod $p$ 
from $(pd_{\mathbf{x},1},pd_{\mathbf{x},2})$, without any
further information on $\mathbf{t}_{\mathbf{x}}$. 
Then, $CA$ computes $\mathsf{fsk}_{\mathbf{x},\mathsf{id},ver} = 
\mathbf{f}_{\mathbf{x},ver} - (\overbrace{ v_{\mathbf{x},\mathsf{id}},\cdots,v_{\mathbf{x},\mathsf{id}}}^{m})$ 
to bind $\mathsf{fsk}_{\mathbf{x},\mathsf{id},ver}$ and $\mathsf{usk}_{\mathsf{id}}$, 
where $v_{\mathbf{x},\mathsf{id}}=\langle\mathbf{id}, \mathbf{t}_{\mathbf{x}}\rangle$ mod $p$. 

When encrypting EHR data $\mathbf{y} \in \mathcal{Y}$, 
$EH$ uses $\mathsf{mpk}$ and $\mathsf{gpk}_{ver}$ to compute the ciphertexts
$\mathsf{CT}_{ver}=(\mathbf{c}_{ver,1},\mathbf{c}_{ver,2})$, where 
$\mathbf{c}_{ver,1} = \mathbf{A}^{\top} \cdot \mathbf{s}_2 +\mathbf{e}_{3} \in \mathbb{Z}_{p}^{m} , 
    \mathbf{c}_{ver,2} = (\mathbf{C}+\mathbf{U}_{ver})^{\top} \cdot\mathbf{s}_2 + \mathbf{e}_4
    + \left\lfloor \frac{p}{K} \right\rfloor \cdot \mathbf{y} \in \mathbb{Z}_{p}^{l_1}$. 
Then, $EH$ uploads the encrypted EHR data $\mathsf{CT}_{ver}$ to $CS$. 

To perform data mining on EHR data, 
a $MI$ obtains $v_{\mathbf{x},\mathsf{id}} = \langle \mathbf{id},\mathbf{t}_{\mathbf{x}} \rangle$ mod $p$ from 
$(pd_{\mathbf{x},1},pd_{\mathbf{x},2})$ using $\mathsf{usk}_{\mathsf{id}}$, 
and then computes $\mathbf{f}_{\mathbf{x},ver} = \mathbf{f}_{\mathbf{x},\mathsf{id},ver} + (\overbrace{v_{\mathbf{x},\mathsf{id}},\cdots, v_{\mathbf{x},\mathsf{id}}}^{m})^{\top}$.
Using $\mathbf{f}_{\mathbf{x},ver}$, 
he/she can obtain $\langle \mathbf{x},\mathbf{y} \rangle$ by decrypting encrypted EHR data $\mathsf{CT}_{ver}$. 
Notably, the inner product can be calculated only if the ciphertext and decryption key include the same version. 

When revocation is required, $GM$ runs $\mathbf{GroupUpdate}$ to generate an updated group secret key
$\mathsf{gsk}_{ver+1} = (\mathbf{D},\mathbf{B}_{ver+1})$, and an updated group public key $\mathsf{gpk}_{ver+1} = (\mathbf{F},\mathbf{U}_{ver+1})$. 
Moreover, $GM$ sets $ver =ver+1$. 

To generate an update key for $CS$, 
$CA$ uses $\mathbf{T}_{\mathbf{A}}$ to sample $\mathbf{Z}_{ver} \in \mathcal{D}_{\mathbb{Z}^{m\times l_1},\rho_1}$ such 
that $\mathbf{C}+\mathbf{U}_{ver} = \mathbf{A}\cdot \mathbf{Z}_{ver}$. 
According to $\mathsf{gpk}_{ver+1}$, $CA$ then generates a update 
key $\mathsf{uptk}_{ver+1}=$
        \begin{equation*} 
            \begin{bmatrix} 
            \mathbf{E}_{1}\mathbf{A}+\mathbf{E}_{2} & \ \ \mathbf{E}_{1}\cdot  (\mathbf{C}+\mathbf{U}_{ver+1}) + \mathbf{E}_{3}-\mathsf{PowerT}_{p}(\mathbf{Z}_{ver})\\ 
           \mathbf{0}_{l_1\times m} & \ \ \mathbf{I}_{l_1\times l_1} \\
            \end{bmatrix}
        \end{equation*}
        for $CS$.

After obtaining $\mathsf{uptk}_{ver+1}$,  
$CS$ can use $\mathsf{uptk}_{ver+1}$ to update encrytped EHR data $\mathsf{CT}_{ver}$ to $\mathsf{CT}_{ver+1}$.
The form of $\mathsf{CT}_{ver+1}$ is
$(\mathbf{c}_{ver+1,1} = \mathbf{A}^{\top} \cdot \mathbf{s}' + error, 
    \mathbf{c}_{ver+1,2} = (\mathbf{C}+\mathbf{U}_{ver+1})^{\top} \cdot\mathbf{s}' + error
    + \left\lfloor \frac{p}{K} \right\rfloor \cdot \mathbf{y})$.  

To broadcast function update information about function key associated with $\mathbf{x}$ for unrevoked $MI$s whose $\mathsf{id} \notin \mathcal{R}_{\mathbf{x}}$ on version $ver+1$, 
$GM$ computes a $\mathbf{v}_{\mathcal{R}_{\mathbf{x}}} \in \mathbb{Z}_p^{l_2}$ such that 
$\langle \mathbf{id},\mathbf{v}_{\mathcal{R}_{\mathbf{x}}} \rangle = 0$ (mod $p$) 
for every $\mathsf{id}\in \mathcal{R}_{\mathbf{x}}$. 
Then, $GM$ computes and broadcasts
$\mathsf{UPI}_{\mathbf{x},ver+1}=(\mathsf{upi}_{\mathbf{x},ver+1,1},\mathsf{upi}_{\mathbf{x},ver+1,2},
\mathsf{upi}_{\mathbf{x},ver+1,3},\mathbf{v}_{\mathcal{R}_\mathbf{x}})$. 
Specifically, $\mathsf{upi}_{\mathbf{x},ver+1,1},\mathsf{upi}_{\mathbf{x},ver+1,2}$ 
are the ciphertexts of $k_t \cdot \mathbf{v}_{\mathcal{R}_{\mathbf{x}}}$, 
where $k_t \in \mathbb{Z}_p$ is randomly selected. 
$\mathsf{upi}_{\mathbf{x},ver+1,3}$ is the ciphertext of $(\mathbf{B}_{ver+1} -\mathbf{B}_{ver}) \cdot \mathbf{x}$, which is encrypted by $k_t$.

In order to update function key associated with $\mathbf{x}$, $MI$s whose 
$\mathsf{id} \notin \mathcal{R}_{\mathbf{x}}$ can use $\mathsf{usk}_{\mathsf{id}}$ to obtain 
the key $k_t$ from $(\mathsf{upi}_{\mathbf{x},ver+1,1},\mathsf{upi}_{\mathbf{x},ver+1,2})$. 
Furthermore, they can obtain the $\mathsf{fsk}_{\mathbf{x},\mathsf{id},ver+1}
=(\mathbf{\mathbf{x}},\mathbf{f}_{\mathbf{x},\mathsf{id},ver+1}=
\mathbf{f}_{\mathbf{x},\mathsf{id},ver} + (\mathbf{B}_{ver+1}-\mathbf{B}_{ver}) \cdot \mathbf{x})$, 
where $(\mathbf{B}_{ver+1}-\mathbf{B}_{ver}) \cdot \mathbf{x}$ can be obtaining by decrypting 
$\mathsf{upi}_{\mathbf{x},ver+1.3}$ using $k_t$.
Because $\mathbf{A} \cdot (\mathbf{Z}_{ver} +(\mathbf{B}_{ver+1} - \mathbf{B}_{ver})) 
= \mathbf{C} + \mathbf{U}_{ver} + \mathbf{U}_{ver+1} - \mathbf{U}_{ver}
= \mathbf{C} +\mathbf{U}_{ver+1}$, 
$\mathsf{fsk}_{\mathbf{x},\mathsf{id},ver+1}$ can decrypt the encrypted EHR data $\mathsf{CT}_{ver+1}$. 
\begin{figure*}
\framebox[\textwidth]{
\parbox{0.97\textwidth}{
    $\mathbf{SystemSetup}$. 
    $CA$ selects the security parameter $1^\lambda$, the bound of $MI$s $\mathcal{N}$, 
    a prime $p$, integers $n,m,l_1,l_2 = \mathcal{N} + 1,X,Y,K=l_1XY,q=p^{k}$ for integer $k \geq 2$, 
    reals $\rho_1,\rho_2, \sigma_1, \sigma_2 > 0$, 
    and spaces $\mathcal{X} = \{0,\ldots,X-1\}^{l_1}$, $\mathcal{Y} = \{0,\ldots,Y-1\}^{l_1}$. 
    Suppose that  
    $H_1: \{0,1\}^* \to \mathbb{Z}_p^{l_2}$, $H_2: \mathbb{Z}_p \to \{ 0,1\} ^ t $ are cryptographic hash functions, and 
    $\mathsf{PRF}:\mathbb{Z}_p \times \mathcal{X} \to \mathbb{Z}_p^{l_2}$ is a 
            pseudorandom function, 
            where 
            $t = m \cdot\left\lceil \mathrm{log\ }(2Xl_1\rho_2) \right\rceil$.
    Then, $CA$ runs $(\mathbf{A},\mathbf{T}_{\mathbf{A}}) \gets \mathsf{TrapGen}(1^n,1^m,p)$ where $\mathbf{A} \in \mathbb{Z}^{n\times m}_p$, 
    and selects 
     $\mathbf{V}\in \mathbb{Z}^{n\times m}_q, \mathbf{C}\gets \mathbb{Z}_p^{n\times l_1}, k_{p} \gets \mathbb{Z}_p$. 
    Let $\mathsf{pd}$ be an initially empty public directory, $\mathsf{msk} = (\mathbf{T}_{\mathbf{A}},k_p)$, 
    $\mathsf{mpk}=(\mathbf{A},\mathbf{V},\mathbf{C})$, $\mathsf{pp}=(n,m,l_1,l_2,p,q,k,\rho_1,\rho_2,\sigma_1,\sigma_2,\mathsf{pd},\mathsf{PRF},H_1,H_2)$.
    \medskip

    $\mathbf{GroupSetup}$.
    $GM$ samples $\mathbf{B}_1 \gets \mathcal{D}_{\mathbb{Z}^{m\times l_1},\rho_1}, \mathbf{D} \gets \mathcal{D}_{\mathbb{Z}^{m\times l_2},\rho_2}$, 
            computes $\mathbf{U}_1 = \mathbf{A} \cdot \mathbf{B}_1 \in \mathbb{Z}_p^{n\times l_1} $,  
            $\mathbf{F} = \mathbf{V} \cdot \mathbf{D} \in \mathbb{Z}_q^{n\times l_2}$,  
            and sets system version number $ver$ as $1$.
            Let $\mathsf{guk} = \mathbf{D}$, $\mathsf{gfk}_1 = \mathbf{B}_1$,
            $\mathsf{gsk}_1 = (\mathsf{guk},\mathsf{gfk}_1)$, and  $\mathsf{gpk}_1 = (\mathbf{F},\mathbf{U}_1)$.
    \medskip
    
    $\mathbf{UKeyGen}$. 
    Given an $MI$ identity $\mathsf{id}$, 
    $GM$ 
            computes $\mathbf{id} = H_1(\mathsf{id}) \in \mathbb{Z}^{l_2}_p$, 
            $\mathbf{u}_{\mathsf{id}}=\mathbf{D}\cdot \mathbf{id} \in \mathbb{Z}^{m}$, and sends user key  
            $\mathsf{usk}_{\mathsf{id}}= (\mathbf{id}, \mathbf{u}_{\mathsf{id}})$ to the $MI$. 
    \medskip
    
    $\mathbf{FKeyGen}$.  
    Given a vector $\mathbf{x} \in \mathcal{X}$ and 
    an $MI$ identity $\mathsf{id}$, 
    $CA$
        samples $\mathbf{Z}_{ver} \gets \mathsf{SamplePre}(\mathbf{A},\mathbf{T}_{\mathbf{A}}, \mathbf{C}+\mathbf{U}_{ver},\rho_1)$, and computes
        $\mathbf{f}_{\mathbf{x},ver} = \mathbf{Z}_{ver} \cdot \mathbf{x} \in \mathbb{Z}^{m} $,  
        $\mathbf{t}_{\mathbf{x}} = \mathsf{PRF}(k_{p},\mathbf{x}) \in \mathbb{Z}^{l_2}_p$, 
        $\mathbf{id} = H_1(\mathsf{id}) \in \mathbb{Z}^{l_2}_p$, 
        $\mathbf{f}_{\mathbf{x},\mathsf{id},ver} = \mathbf{f}_{\mathbf{x},ver} - 
        (\overbrace{ v_{\mathbf{x},\mathsf{id}},\cdots, v_{\mathbf{x},\mathsf{id}}}^{m})^{\top} \in \mathbb{Z}^m$, where 
        $v_{\mathbf{x},\mathsf{id}} = \langle \mathbf{id}, \mathbf{t}_{\mathbf{x}}\rangle$ mod $p$.
        Let the public directory $\mathsf{pd}$ contain $(p_{\mathbf{x},1},p_{\mathbf{x},2})$, 
        where $\mathbf{x}$ is the function key queries that have been made so far. 
        If $(pd_{\mathbf{x},1},pd_{\mathbf{x},2}) \notin \mathsf{pd}$,  
        $CA$ samples $\mathbf{s}_1 \gets \mathbb{Z}_{q}^{n}$, 
        $\mathbf{e}_{1} \gets \mathcal{D}_{\mathbb{Z},\sigma_1}^{m}$, 
        $\mathbf{e}_{2} \gets \mathcal{D}_{\mathbb{Z},\sigma_1}^{l_2}$, 
        computes $pd_{\mathbf{x},1} = \mathbf{V}^{\top} \cdot \mathbf{s}_1 + \mathbf{e}_{1} \in \mathbb{Z}_{q}^{m}$
        $pd_{\mathbf{x},2} = \mathbf{F}^{\top} \cdot \mathbf{s}_1 + \mathbf{e}_{2} + p^{k-1} \cdot \mathbf{t}_{\mathbf{x}} \in \mathbb{Z}_{q}^{l_2}$, 
        and appends $(pd_{\mathbf{x},1},pd_{\mathbf{x},2})$ to $\mathsf{pd}$.
        $CA$ sends function key
        $\mathsf{fsk}_{\mathbf{x},\mathsf{id},ver} = (\mathbf{x},\mathbf{f}_{\mathbf{x},\mathsf{id},ver})$ to the $MI$. 
    \medskip
    
    $\mathbf{Enc}$. 
    To encrypt EHR data 
    $\mathbf{y}\in \mathcal{Y}$, 
    $EH$ samples $\mathbf{s}_2 \gets \mathbb{Z}_{p}^{n}$, $\mathbf{e}_{3} \gets \mathcal{D}_{\mathbb{Z},\sigma_2}^{m}$, 
    $\mathbf{e}_{4} \gets \mathcal{D}_{\mathbb{Z},\sigma_2}^{l_1}$, 
    computes $\mathbf{c}_{ver,1} = \mathbf{A}^{\top} \cdot \mathbf{s}_2 +\mathbf{e}_{3} \in \mathbb{Z}_{p}^{m} , 
    \mathbf{c}_{ver,2} = (\mathbf{C}+\mathbf{U}_{ver})^{\top} \cdot\mathbf{s}_2 + \mathbf{e}_4
    + \left\lfloor \frac{p}{K} \right\rfloor \cdot \mathbf{y} \in \mathbb{Z}_{p}^{l_1}$, 
    and uploads encrypted EHR data $\mathsf{CT}_{ver}=(\mathbf{c}_{ver,1},\mathbf{c}_{ver,2})$ to $CS$. 
    \medskip
    
    $\mathbf{Dec}$. 
    To decrypt the ciphertext of EHR data,  
    $\mathsf{CT}_{ver}$, 
    $EH$
    computes $\theta '=  {\mathbf{id}}^{\top} \cdot pd_{\mathbf{x},2} -  
    \mathbf{u}_{\mathsf{id}}^{\top} \cdot pd_{\mathbf{x},1} $ mod $q$ 
    and obtains $\theta \in \mathbb{Z}_{p}$, where $\theta$  is the value that minimizes $\lvert p^{k-1} \cdot \theta - \theta' \rvert$. 
    Note that $\theta =\langle \mathbf{id}, \mathbf{t}_{\mathbf{x}}\rangle$ mod $p$. 
    Then, $EH$
    computes $\mathbf{f}_{\mathbf{x},ver} = \mathbf{f}_{\mathbf{x},\mathsf{id},ver} + (\overbrace{ \theta,\cdots, \theta}^{m})^{\top} \in \mathbb{Z}^m$, 
    $\mu' = {\mathbf{x}}^{\top} \cdot \mathbf{c}_{ver,2} -  
    {\mathbf{f}_{\mathbf{x},ver}^{\top}} \cdot \mathbf{c}_{ver,1} $ mod $p$, and 
    outputs $\mu \in \{0,\ldots,K-1 \}$, where $\mu$ is the value that minimizes $\lvert \left\lfloor \frac{p}{K} \right\rfloor  \cdot \mu - \mu' \rvert$. 
    \medskip
    
    $\mathbf{GroupUpdate}$. 
    $GM$
    samples $\mathbf{B}_{ver+1} \gets \mathcal{D}_{\mathbb{Z}^{m\times l_1},\rho_1}$, and 
    computes $\mathbf{U}_{ver+1} = \mathbf{A} \cdot \mathbf{B}_{ver+1} \in \mathbb{Z}_p^{n\times l_1}$.
    Let $\mathsf{gsk}_{ver+1} = (\mathbf{D},\mathbf{B}_{ver+1})$, $\mathsf{gpk}_{ver+1} = (\mathbf{F},\mathbf{U}_{ver+1})$, and $ver=ver+1$.
    \medskip

    $\mathbf{UptKeyGen}$.  
    $CA$ samples $\mathbf{Z}_{ver} \gets \mathsf{SamplePre}(\mathbf{A},\mathbf{T}_{\mathbf{A}}, \mathbf{C}+\mathbf{U}_{ver},\rho_1)$, 
    $\mathbf{E}_{1} \gets \mathbb{Z}_{p}^{hm \times n}$, 
    $\mathbf{E}_{2} \gets \mathcal{D}_{\mathbb{Z}^{hm\times m},\sigma_2}$, 
    and $\mathbf{E}_{3} \gets \mathcal{D}_{\mathbb{Z}^{hm\times l_1},\sigma_2}$, where $h=\lceil \mathrm{log\ }p \rceil$. 
    Then, $CA$
    computes and sends update key $\mathsf{uptk}_{ver+1}=$
        \begin{equation*} 
            \begin{bmatrix} 
            \mathbf{E}_{1}\mathbf{A}+\mathbf{E}_{2} & \ \ \mathbf{E}_{1}\cdot  (\mathbf{C}+\mathbf{U}_{ver+1}) + \mathbf{E}_{3}-\mathsf{PowerT}_{p}(\mathbf{Z}_{ver})\\ 
           \mathbf{0}_{l_1\times m} & \ \ \mathbf{I}_{l_1\times l_1} \\
            \end{bmatrix}
        \end{equation*}
    to $CS$. 
    \medskip

    $\mathbf{CTUpdate}$. 
    To update the ciphertext $\mathsf{CT}_{ver}$ of EHR data to a new version, 
    $CS$ 
    computes $(\mathbf{c}^{\top}_{ver+1,1},\mathbf{c}^{\top}_{ver+1,2})=(\mathsf{BitD}_{p}(\mathbf{c}_{ver,1}^{\top}),\mathbf{c}_{ver,2}^{\top}) \cdot \mathsf{uptk}_{ver+1}$, 
    and outputs updated encrypted EHR data $\mathsf{CT}_{ver+1}=(\mathbf{c}_{ver+1,1},\mathbf{c}_{ver+1,2})$. 
    \medskip
    
    $\mathbf{FUpdate}$. 
   To broadcast the function update information about  
    a vector $\mathbf{x}\in \mathcal{X}$ 
    to unrevoked $MI$s, 
    $GM$ 
        computes $\mathbf{v}_{\mathcal{R}_{\mathbf{x}}}\in \mathbb{Z}^{l_2}_{p}\setminus  \{\mathbf{0}\}$ such that 
        $\langle \mathbf{id},\mathbf{v}_{\mathcal{R}_{\mathbf{x}}} \rangle = 0$ (mod $p$) for every $\mathsf{id} \in  \mathcal{R}_{\mathbf{x}}$, 
        and samples $\mathbf{s}_3 \gets \mathbb{Z}_{q}^{n}$, 
        $\mathbf{e}_{5} \gets \mathcal{D}_{\mathbb{Z},\sigma_1}^{m}$, 
        $\mathbf{e}_{6} \gets \mathcal{D}_{\mathbb{Z},\sigma_1}^{l_2}$ 
        $k_t \gets \mathbb{Z}_p$. 
        Then, $GM$
        computes
        $\mathsf{upi}_{\mathbf{x},ver+1,1}=\mathbf{V}^{\top} \cdot \mathbf{s}_3 + \mathbf{e}_{5} \in \mathbb{Z}_{q}^{m}$, 
        $\mathsf{upi}_{\mathbf{x},ver+1,2}=\mathbf{F}^{\top} \cdot \mathbf{s}_3 + \mathbf{e}_{6} + p^{k-1} \cdot (k_t \cdot \mathbf{v}_{\mathcal{R}_{\mathbf{x}}}) \in \mathbb{Z}_{q}^{l_2}$ and 
        $\mathsf{upi}_{\mathbf{x},ver+1,3} = H_2(k_t) \oplus \mathsf{bin}( (\mathbf{B}_{ver+1}-\mathbf{B}_{ver}) \cdot \mathbf{x}) \in \{0,1\}^t$, 
        and broadcasts function update information
        $\mathsf{UPI}_{\mathbf{x},ver+1}=(\mathsf{upi}_{\mathbf{x},ver+1,1},\mathsf{upi}_{\mathbf{x},ver+1,2},\mathsf{upi}_{\mathbf{x},ver+1,3},\mathbf{v}_{\mathcal{R}_\mathbf{x}})$
        to $MI$s. 
    \medskip
    
    $\mathbf{KeyUpdate}$. 
    A unrevoked $MI$ with function key $\mathsf{fsk}_{\mathbf{x},\mathsf{id},ver}$
    computes $\nu'= {\mathbf{id}}^{\top}\cdot \mathsf{upi}_{\mathbf{x},ver+1,2} -  
            \mathbf{u}_{\mathsf{id}}^{\top} \cdot \mathsf{upi}_{\mathbf{x},ver+1,1} $ mod $q$ and 
            $k_t = \frac{\nu}{\langle \mathbf{id}, \mathbf{v}_{\mathcal{R}_\mathbf{x}} \rangle} 
            = \frac{\langle \mathbf{id}, k_t \cdot \mathbf{v}_{\mathcal{R}_\mathbf{x}} \rangle}
            {\langle \mathbf{id}, \mathbf{v}_{\mathcal{R}_\mathbf{x}} \rangle}$ mod $p$, 
            where $\nu \in \mathbb{Z}_{p}$ is the value that minimizes $\lvert p^{k-1} \cdot \nu - \nu' \rvert$. 
            Then, $MI$
            computes $\mathsf{bin}((\mathbf{B}_{ver+1}-\mathbf{B}_{ver}) \cdot \mathbf{x})
            = H_2(k_t)\oplus \mathsf{upi}_{\mathbf{x},ver+1,3}$, 
            $\mathbf{f}_{\mathbf{x},\mathsf{id},ver+1}=
            \mathbf{f}_{\mathbf{x},\mathsf{id},ver} + (\mathbf{B}_{ver+1}-\mathbf{B}_{ver}) \cdot \mathbf{x} \in \mathbb{Z}^m$,
            and obtains the new version function key $\mathsf{fsk}_{\mathbf{x},\mathsf{id},ver+1}=
            (\mathbf{x},\mathbf{f}_{\mathbf{x},\mathsf{id},ver+1})$. 
    \medskip
}}
\caption{The concrete construction of lattice-based IPFE-FR}\label{figure:CC}
\end{figure*}

$\mathbf{Correctness.}$ 
The correctness of $\mathbf{Dec}$ can be shown as follows.
In $\mathbf{Dec}$, 
we have 
$\theta' = {\mathbf{id}^{\top}}\cdot pd_{\mathbf{x},2} - 
                {\mathbf{u}_{\mathsf{id}}^{\top}}\cdot pd_{\mathbf{x},1} 
                ={\mathbf{id}^{\top}}\cdot (\mathbf{F}^{\top} \cdot \mathbf{s}_1 + \mathbf{e}_{2} + p^{k-1} \cdot \mathbf{t}_{\mathbf{x}})-
                {(\mathbf{D}\cdot \mathbf{id})}^{\top} \cdot (\mathbf{V}^{\top} \cdot \mathbf{s}_1 + \mathbf{e}_{1})
                = p^{k-1} \cdot (\langle \mathbf{id},\mathbf{t}_{\mathbf{x}} \rangle\mathrm{\ mod\ } p)+
                {\mathbf{id}}^{\top}\cdot \mathbf{e}_2 -
                {\mathbf{id}}^{\top} \mathbf{D}^{\top} \cdot \mathbf{e}_{1} \mathrm{\ mod\ } q, $
relying on the fact that $\mathbf{F} = \mathbf{V}\cdot \mathbf{D} \in \mathbb{Z}_q^{n\times l_2}$. 

The above step is correct as long as 
$
\lvert
    {{\mathbf{id}^{\top}}\cdot \mathbf{e}_2
    - \mathbf{id}}^{\top} {\mathbf{D}}^{\top} \cdot \mathbf{e}_{1}
\rvert \leq \frac{p^{k-1}}{4}.$
Since $\mathbf{D} \gets \mathcal{D}_{\mathbb{Z}^{m\times l_2},\rho_2}$, 
$\mathbf{e}_{1} \gets \mathcal{D}_{\mathbb{Z}^{m},\sigma_1}$, 
$\mathbf{e}_{2} \gets \mathcal{D}_{\mathbb{Z}^{l_2},\sigma_1}$ and $\mathbf{id} \in \mathbb{Z}^{l_2}_p$, 
we know that 
$\left\lVert {\mathbf{D}^{\top}} \right\rVert \leq \sqrt{l_2} \cdot \rho_2$, 
$\left\lVert \mathbf{e}_{1} \right\rVert \leq \sqrt{m} \cdot \sigma_1$, 
$ \left\lVert\mathbf{e}_{2} \right\rVert\leq \sqrt{l_2} \cdot \sigma_1$
and $\left\lVert \mathbf{id} \right\rVert \leq \sqrt{l_2} p $
by Lemma \ref{Bound Gaussian}. 

Then, we have $\lvert \mathbf{id}^{\top} \cdot \mathbf{e}_2 \rvert \leq l_{2} \cdot p \sigma_1$
and $ \lvert \mathbf{id}^{\top} {\mathbf{D}}^{\top} \cdot \mathbf{e}_{1} \rvert 
\leq l_{2} \sqrt{m} \cdot p \rho_2 \sigma_1$. 
To obtain the correct inner product 
$\langle \mathbf{id},\mathbf{t}_{\mathbf{x}}\rangle$ mod $p$, 
the inequation 
$\lvert
    {{\mathbf{id}^{\top}}\cdot \mathbf{e}_2
    - \mathbf{id}}^{\top} {\mathbf{D}}^{\top} \cdot \mathbf{e}_{1}
\rvert \leq
l_{2} p \sigma_1 + l_{2} \sqrt{m} p \rho_2 \sigma_1
\leq \frac{p^{k-1}}{4}$ needs to be satisfied. 

Furthermore, we have 
$\mathbf{f}_{x,ver+1} = \mathbf{Z}_{ver} \cdot \mathbf{x} - (\overbrace{ \theta,\cdots, \theta}^{m})^{\top} +
(\overbrace{ v_{\mathbf{x},\mathsf{id}},\cdots, v_{\mathbf{x},\mathsf{id}}}^{m})^{\top} 
= \mathbf{Z}_{ver} \cdot \mathbf{x}$ 
and
$\mu' = \mathbf{x}^{\top} \cdot \mathbf{c}_{2} - ({\mathbf{Z}_{ver} \cdot \mathbf{x}})^{\top}\cdot \mathbf{c}_{1} 
            = \mathbf{x}^{\top} \cdot ((\mathbf{C}+\mathbf{U}_{ver})^{\top} \cdot\mathbf{s}_2 + \mathbf{e}_4 
            + \left\lfloor \frac{p}{K} \right\rfloor \cdot \mathbf{y})-
            (\mathbf{Z}_{ver} \cdot \mathbf{x})^{\top}\cdot (\mathbf{A}^{\top} \cdot \mathbf{s}_2 +\mathbf{e}_{3})
            = \left\lfloor \frac{p}{K} \right\rfloor \cdot \langle \mathbf{x},\mathbf{y} \rangle + 
            \mathbf{x}^{\top} \cdot \mathbf{e}_{4}
            - \mathbf{x}^{\top} {\mathbf{Z}_{ver}}^{\top} \cdot \mathbf{e}_{3} \mathrm{\ mod\ } p, $
where we use the fact that $\mathbf{C}+\mathbf{U}_{ver} = \mathbf{A}\cdot \mathbf{Z}_{ver} \in \mathbb{Z}_p^{n\times l_1}$. 

Since $\mathbf{Z}_{ver} \in \mathbb{Z}^{m\times l_{1}}$ is sampled from $\mathsf{SamplePre}$ algorithm with parameter $\rho_1$, 
$\mathbf{e}_{3} \gets \mathcal{D}_{\mathbb{Z}^{m},\sigma_2}$, 
$\mathbf{e}_{4} \gets \mathcal{D}_{\mathbb{Z}^{l_1},\sigma_2}$ and $\mathbf{x} \in \mathcal{X}$, 
we know that 
$\left\lVert \mathbf{Z}_{ver}^{\top} \right\rVert \leq \sqrt{l_1} \cdot \rho_1$, 
$ \left\lVert \mathbf{e}_{3} \right\rVert \leq \sqrt{m} \cdot \sigma_2$ 
, $\left\lVert \mathbf{e}_{4} \right\rVert \leq \sqrt{l_1} \cdot \sigma_2$  
and $\left\lVert \mathbf{x} \right\rVert \leq \sqrt{l_1} X$
by Lemma \ref{Bound Gaussian}. 

Then, we have $\lvert \mathbf{x}^{\top} \cdot \mathbf{e}_4 \rvert \leq l_{1} \cdot X \sigma_2$
and $ \lvert \mathbf{x}^{\top} {\mathbf{Z}_{ver}}^{\top} \cdot \mathbf{e}_{3} \rvert 
\leq l_{1} \sqrt{m} \cdot X \rho_1 \sigma_2$. 
To ensure the correctness of the computed inner product  
$\langle \mathbf{x},\mathbf{y}\rangle$, 
the inequation 
$\lvert
        \mathbf{x}^{\top} \cdot \mathbf{e}_{4} 
        - \mathbf{x}^{\top} {\mathbf{Z}_{ver}}^{\top} \cdot \mathbf{e}_{3}
    \rvert \leq
l_{1} X \sigma_2 + l_{1} \sqrt{m} X \rho_1 \sigma_2
\leq \frac{p}{4K}$ 
needs to be satisfied.

The correctness of fine-grained revocation can be shown as follows. 
In $\mathbf{UptKeyGen}$ and $\mathbf{CTUpdate}$, 
we have 
$(\mathsf{BitD}_{p}(\mathbf{c}_{ver,1}^{\top}),\mathbf{c}_{ver,2}^{\top}) \cdot \mathsf{uptk}_{ver+1}
    =(\mathsf{BitD}_{p}(\mathbf{c}_{ver,1}^{\top}),\mathbf{c}_{ver,2}^{\top}) \cdot
    \begin{bmatrix} 
        \mathbf{E}_{1}\mathbf{A}+\mathbf{E}_{2} & \ \ \mathbf{E}_{1} \cdot (\mathbf{C}+\mathbf{U}_{ver+1}) + \mathbf{E}_{3}-\mathsf{PowerT}_{p}(\mathbf{Z}_{ver})\\ 
        \mathbf{0}_{l\times m} & \ \ \mathbf{I}_{l\times l} \\
    \end{bmatrix}\\
    = (
        \mathsf{BitD}_{p}(\mathbf{c}_{ver,1}^{\top}) \cdot (\mathbf{E}_{1}\mathbf{A}+\mathbf{E}_{2}), 
        \mathsf{BitD}_{p}(\mathbf{c}_{ver,1}^{\top}) \cdot 
    (\mathbf{E}_{1} \cdot (\mathbf{C}+\mathbf{U}_{ver+1}) + \mathbf{E}_{3}-\mathsf{PowerT}_{p}(\mathbf{Z}_{ver})) + \mathbf{c}_{ver,2}^{\top})
    = (
        \mathsf{BitD}_{p}(\mathbf{c}_{ver,1}^{\top}) \cdot (\mathbf{E}_{1}\mathbf{A}+\mathbf{E}_{2}),
        \mathsf{BitD}_{p}(\mathbf{c}_{ver,1}^{\top}) \cdot  
        (\mathbf{E}_{1}\cdot (\mathbf{C}+\mathbf{U}_{ver+1})+ \mathbf{E}_{3})- 
        \mathbf{e}^{\top}_{3} \cdot \mathbf{Z}_{ver} + \mathbf{e}_{4}^{\top} + \left\lfloor \frac{p}{K} \right\rfloor \cdot \mathbf{y}^{\top}).$

Hence, we have 
$\mathbf{c}'_{ver,1} = \mathbf{A}^{\top} \mathbf{s}'+\mathbf{e}'_{3}, 
\mathbf{c}'_{ver,2} = (\mathbf{C}+\mathbf{U}_{ver+1})^{\top} \mathbf{s}' + \mathbf{e}'_{4} 
+ \left\lfloor \frac{p}{K} \right\rfloor \cdot \mathbf{y}$ is valid encrypted EHR data,  
where 
$\mathbf{s}' = (\mathsf{BitD}_{p}(\mathbf{c}_{ver,1}^{\top})\mathbf{E}_{1})^{\top}$, 
$\mathbf{e}'_{3} = (\mathsf{BitD}_{p}(\mathbf{c}^{\top}_{ver,1})\mathbf{E}_{2})^{\top}$ and
$\mathbf{e}'_{4} = (\mathsf{BitD}_{p}(\mathbf{c}^{\top}_{ver,1})\mathbf{E}_{3})^{\top} 
- \mathbf{Z}_{ver}^{\top} \cdot \mathbf{e}_{3} + \mathbf{e}_{4}$. 
Moreover, 
in $\mathbf{KeyUpdate}$, we have 
$\nu' = {\mathbf{id}}^{\top}\cdot \mathsf{upi}_{\mathbf{x},ver+1,2} -  
\mathbf{u}_{\mathsf{id}}^{\top} \cdot \mathsf{upi}_{\mathbf{x},ver+1,1} 
= p^{k-1} \cdot (\langle \mathbf{id},k_t \cdot \mathbf{v}_{\mathcal{R}_{\mathbf{x}}} \rangle\mathrm{\ mod\ } p) +
{\mathbf{id}}^{\top}\cdot \mathbf{e}_6 -
{\mathbf{id}}^{\top} {\mathbf{Z}}^{\top} \cdot \mathbf{e}_{5} \mathrm{\ mod\ } q$. 
Then, we can get 
$
\frac{\nu}{\langle \mathbf{id}_i, \mathbf{v}_{\mathcal{R}_\mathbf{x}} \rangle} 
= \frac{\langle \mathbf{id},k_t \cdot \mathbf{v}_{\mathcal{R}_{\mathbf{x}}} \rangle}
{\langle \mathbf{id}, \mathbf{v}_{\mathcal{R}_\mathbf{x}} \rangle} 
=k_t \mathrm{\ mod\ } p$ 
and 
$H_2(k_t)\oplus \mathsf{upi}_{\mathbf{x},ver+1,3}
=\mathsf{bin}( (\mathbf{B}_{ver+1}-\mathbf{B}_{ver}) \cdot \mathbf{x})$. 
Furthermore, we have
$\mathbf{f}_{\mathbf{x},\mathsf{id},ver+1}=
\mathbf{f}_{\mathbf{x},\mathsf{id},ver} + (\mathbf{B}_{ver+1}-\mathbf{B}_{ver}) \cdot \mathbf{x}
= (\mathbf{Z}_{ver}+\mathbf{B}_{ver+1}-\mathbf{B}_{ver}) \cdot \mathbf{x} 
- (\overbrace{ v_{\mathbf{x},\mathsf{id}},\cdots, v_{\mathbf{x},\mathsf{id}}}^{m})^{\top}$ is a valid function key. 

$\mathbf{Setting\ Parameters}.$
We select the parameters $n, m, q, \sigma_1, \sigma_2, \rho_1, \rho_2$ as described in the ALS-IPFE in Section \ref{ALS}. 
We subsequently refine the parameters to maintain our scheme's correctness and security, without affecting the ALS-IPFE's correctness and security.
The parameters of our scheme can be set as follows:
$l_{2} p \sigma_1 + l_{2} \sqrt{m} p \rho_2 \sigma_1
    \leq \frac{p^{k-1}}{4}$,
    $l_{1} X \sigma_2 + l_{1} \sqrt{m} X \rho_1 \sigma_2
    \leq \frac{p}{4K}$, 
    $\alpha p >2\sqrt{n}$,
    $m \geq 2n \mathrm{log}\ p $,
    $\rho_1,\rho_2 \geq \omega(\mathrm{log\ }n)$,
    $\sigma_1 = 2C\alpha q (\sqrt{m}+\sqrt{n}+\sqrt{l_2})$,
    $\sigma_2 = 2C\alpha p (\sqrt{m}+\sqrt{n}+\sqrt{l_1})$,
where $C$ denoting the constant specified in the Lemma \ref{BoundGauss}. 

\section{Performance} \label{section:seventh}
\subsection{Comparison}
This section presents a cost comparison between our IPFE-FR scheme and the related scheme.
Since the schemes \cite{luo2022generic,luo2024fully} support direct revocation, 
we compare the communication cost, storage cost and computational cost with the indirect revocable IPFE scheme \cite{zhu2025revocable}, 
which is the most similar scheme to our IPFE-FR scheme. 

We present a comparison of the communication cost, storage cost and computational cost between our scheme and the indirect revocable IPFE scheme \cite{zhu2025revocable} in Tables \ref{tab:communication_cost}, \ref{tab:storge_cost} and \ref{tab:computation_cost}, respectively. 
In Tables \ref{tab:communication_cost} and \ref{tab:storge_cost}, $l$, $l_d$, and $l_{id}$ represent the dimension of the vectors, the maximum depth of identity, and the depth of identity in \cite{zhu2025revocable}, respectively. For a bilinear map $e: \mathbb{G} \times \mathbb{G} \to \mathbb{G}_T$, let $E_G$, $E_T$, and $E_p$ denote the length of one element in the groups $\mathbb{G}$, $\mathbb{G}_T$, and $\mathbb{Z}_p$, respectively.

In Table. \ref{tab:computation_cost}, the following notations are used: 
$b_T$, $m_T$ and $p_T$ are the time of executing the bilinear pairing operation, the multiplication operation 
and power operation on group $\mathbb{G}_T$, respectively.
$p_G$ and $m_G$ are the time of executing the power operation and the multiplication operation on group $\mathbb{G}$, respectively.
Moreover, $t_M$ is the time of executing the multiplication
operations between matrices, 
$t_G$ is the time of executing the trapdoor generation algorithm, 
$t_S$ is the time of executing the sample preimage function algorithm, 
$t_B$ is time of executing the algorithm $\mathsf{BitD_p}$
and $t_P$ is time of executing the algorithm $\mathsf{PowerT_p}$. 

\begin{table*}[!ht] 
    \caption{The comparison of communication cost} \label{tab:communication_cost}
    \centering
    \normalsize
    \resizebox{0.9\textwidth}{!}{
    \begin{tabular}{|c|c|c|c|c|c|c|}
    \hline
    Scheme&$\mathbf{SystemSetup}$&$\mathbf{GroupSetup}$&$\mathbf{UKeyGen}$&\multicolumn{2}{c|}{$\mathbf{FKeyGen}$}&$\mathbf{Enc}$\\
    \hline
    \cite{zhu2025revocable}&$(l_d+4)E_G+E_T$&$-$&$-$&\multicolumn{2}{c|}{$lE_p + (l_d-l_{id}+2)E_G$}& $lE_T+2E_G$\\
    \hline
    IPFE-FR&$(m+mk+l_1)n\mathrm{log} p$&$(l_{1} + kl_2)n \mathrm{log}p$&$(m+l_2) \mathrm{log} p + m \mathrm{log} (l_2 \rho_2)$&\multicolumn{2}{c|}{$(m+l_1) \mathrm{log}X + m\mathrm{log} (l_1\rho_1) + (km+kl_2)\mathrm{log}p$}&$(m+l_1) \mathrm{log} p$\\
    \hline
    &$\mathbf{Dec}$&$\mathbf{GroupUpdate}$&$\mathbf{UptKeyGen}$& $\mathbf{CTUpdate}$&$\mathbf{FUpdate}$&$\mathbf{KeyUpdate}$\\
    \hline
    \cite{zhu2025revocable}&$-$&$-$&$-$&$-$&$-$&$-$\\
    \hline
    IPFE-FR&$-$&$(l_{1} + kl_2)n \mathrm{log}p$&$(m+l_1) (hm\mathrm{log}p +l_1)$&$-$&$(kl_2+km+l_2)\mathrm{log}p +t$&$-$\\
    \hline
    \end{tabular}}
\end{table*}

\begin{table*}[!ht] 
    \caption{The comparison of storage cost} \label{tab:storge_cost}
    \centering
    \normalsize
    \resizebox{0.9\textwidth}{!}{
    \begin{tabular}{|c|c|c|c|c|c|}
    \hline
    Scheme&master secret key&group secret key& user key &function key&update key\\
    \hline
    \cite{zhu2025revocable}& $E_G$&$-$&$-$&$ lE_p + (l_d-l_{id}+2)E_G$&$E_p$\\
    \hline
    IPFE-FR&$nm\mathrm{log} p$&$m (l_{1}\mathrm{log}\rho_1+ l_2\mathrm{log}\rho_2)$&$(m+l_2) \mathrm{log} p + m \mathrm{log} (l_2 \rho_2)$&$(m+l_1) \mathrm{log}X + m\mathrm{log} (l_1\rho_1)$&$(m+l_1) (hm\mathrm{log}p +l_1)$\\
    \hline
    \end{tabular}}
\end{table*}

\begin{table*}[!ht] 
    \caption{The comparison of computational cost} \label{tab:computation_cost}
    \centering
    \normalsize
    \resizebox{0.8\textwidth}{!}{
    \begin{tabular}{|c|c|c|c|c|c|c|}
    \hline
    Scheme&$\mathbf{SystemSetup}$&$\mathbf{GroupSetup}$&$\mathbf{UKeyGen}$&\multicolumn{2}{c|}{$\mathbf{FKeyGen}$}&$\mathbf{Enc}$\\
    \hline
    \cite{zhu2025revocable}&$2p_G+b_t$&$-$&$-$&\multicolumn{2}{c|}{$(l+2)p_G+lm_G$}&$(l+1)p_G+lm_G+2p_T+m_T$\\
    \hline
    IPFE-FR&$t_G$&$2t_M$&$t_M$&\multicolumn{2}{c|}{$t_S+4t_M$}&$2t_M$\\
    \hline
    &$\mathbf{Dec}$&$\mathbf{GroupUpdate}$&$\mathbf{UptKeyGen}$& $ \mathbf{CTUpdate}$&$\mathbf{FUpdate}$&$\mathbf{KeyUpdate}$\\
    \hline
    \cite{zhu2025revocable}&$p_T+2m_T+2b_T$&$-$&$-$&$2p_G+p_T$&$-$&$-$\\
    \hline
    IPFE-FR&$4t_M$&$t_M$&$t_S+2t_M+t_P$&$t_B+t_M$&$3t_M$&$3t_M$\\
    \hline
    \end{tabular}}
\end{table*}

\subsection{Implement} 
In this section, we implement and evaluate our lattice-based IPFE-FR scheme using SageMath, with the parameter set to 
$n=64$. 
All experiments are performed on a Dell Inspiron 13-5320 running Ubuntu 22.04.4 LTS, powered by a 12th Gen Intel Core i5-1240P processor, with 16 GB RAM and a 512 GB solid-state drive.
The computational cost of our IPFE-FR scheme is presented in Fig. \ref{figure:test_time}.

\begin{figure*}[!ht]
    \centering
        \subfloat[$\mathbf{SystemSetup}$]{\label{figure:test_SystemSetup}\includegraphics[width = 0.25\textwidth]{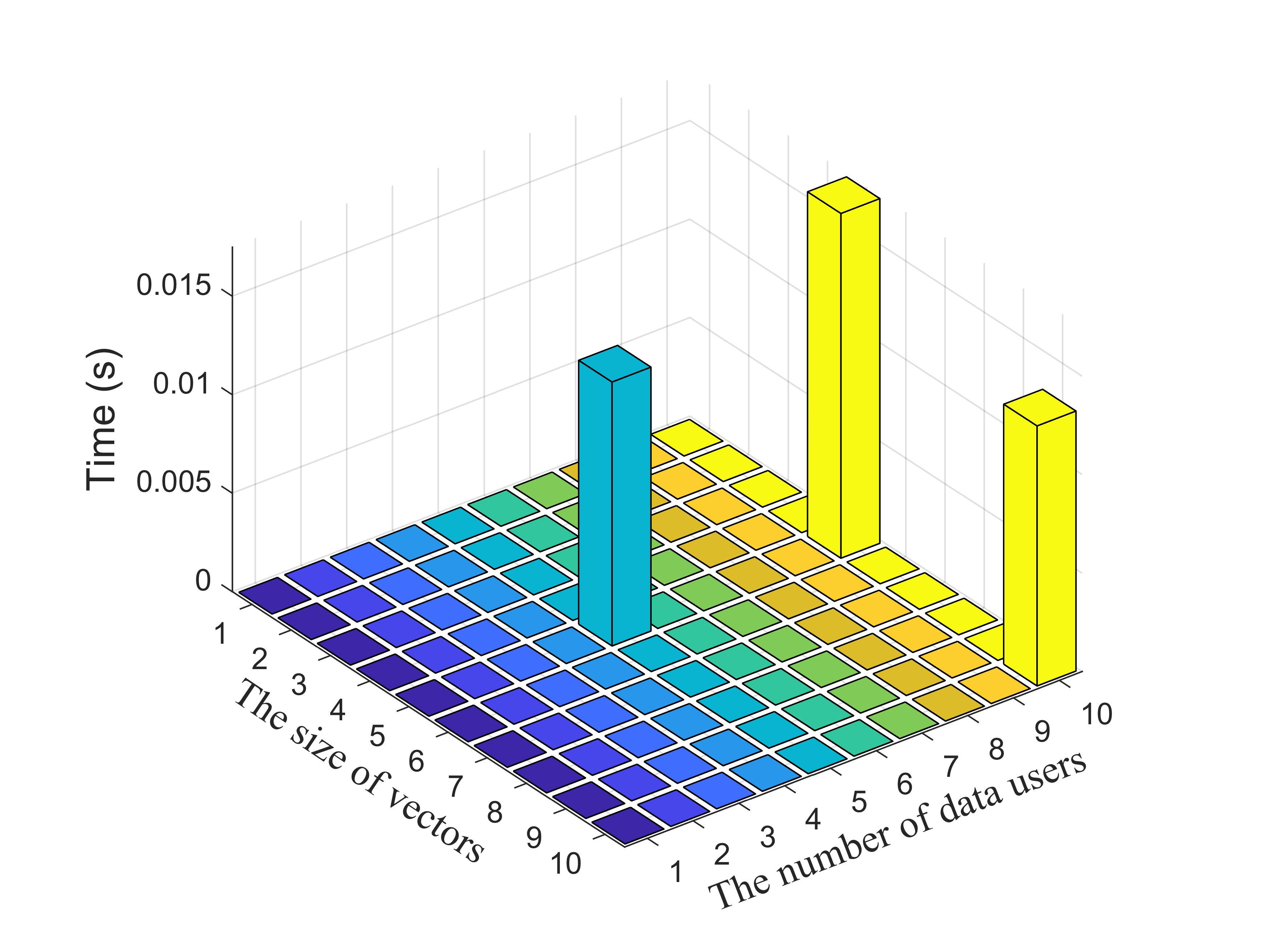}}
        \subfloat[$\mathbf{GroupSetup}$]{\label{figure:test_GroupSetup}\includegraphics[width = 0.25\textwidth]{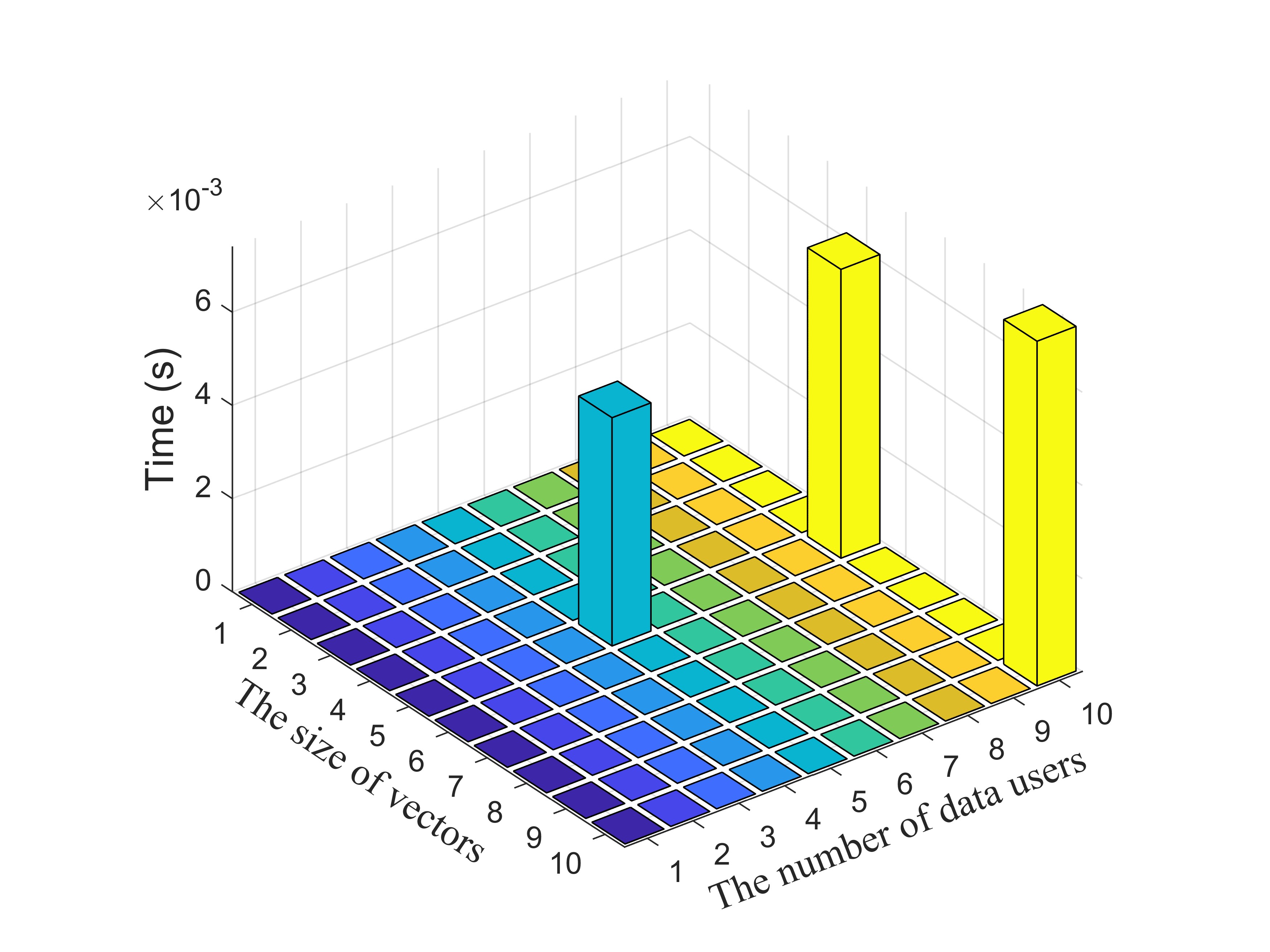}}
        \subfloat[$\mathbf{UKeyGen}$]{\label{figure:test_UKeyGen}\includegraphics[width = 0.25\textwidth]{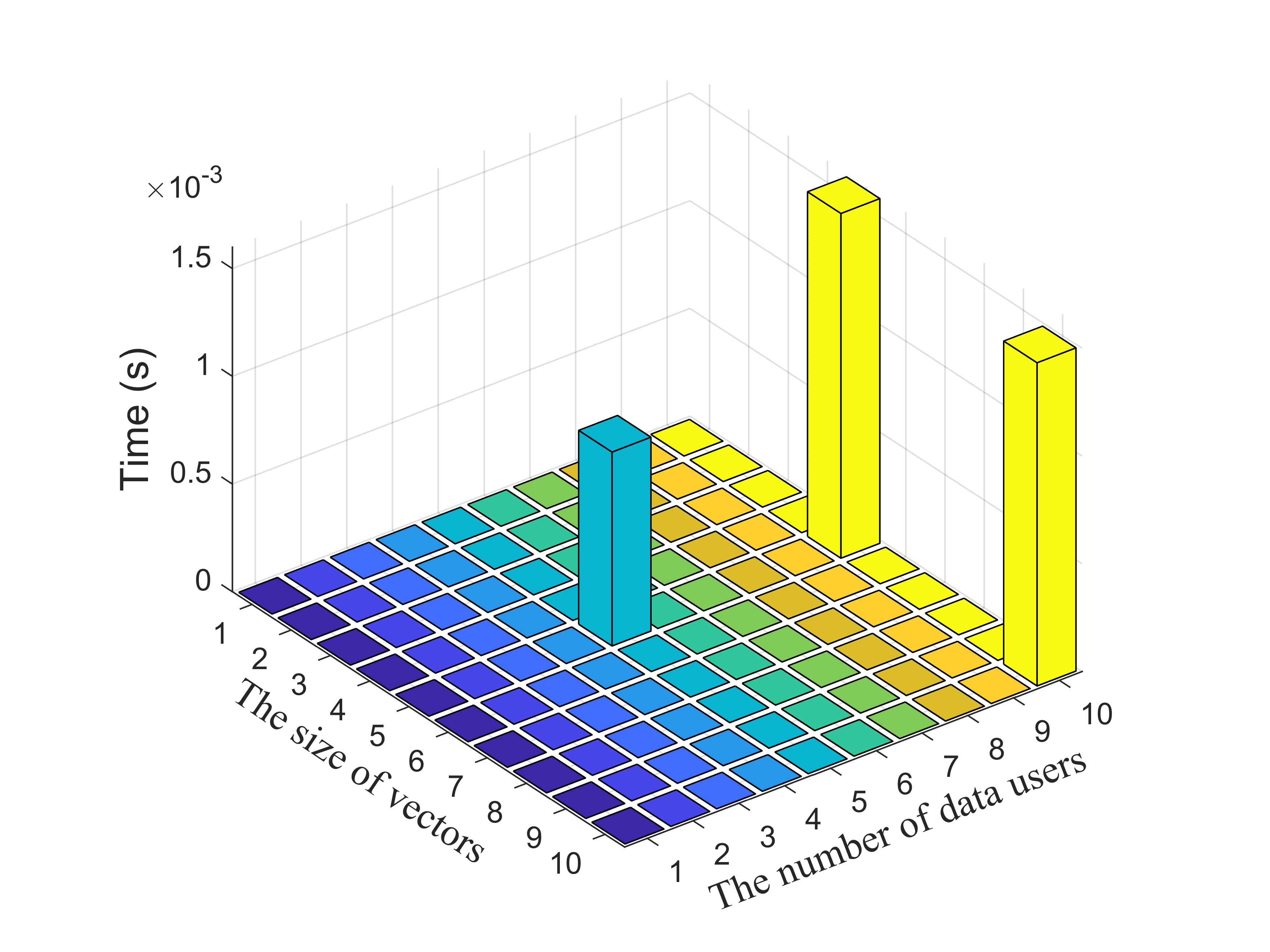}}
        \subfloat[$\mathbf{FKeyGen}$]{\label{figure:test_FKeyGen}\includegraphics[width = 0.25\textwidth]{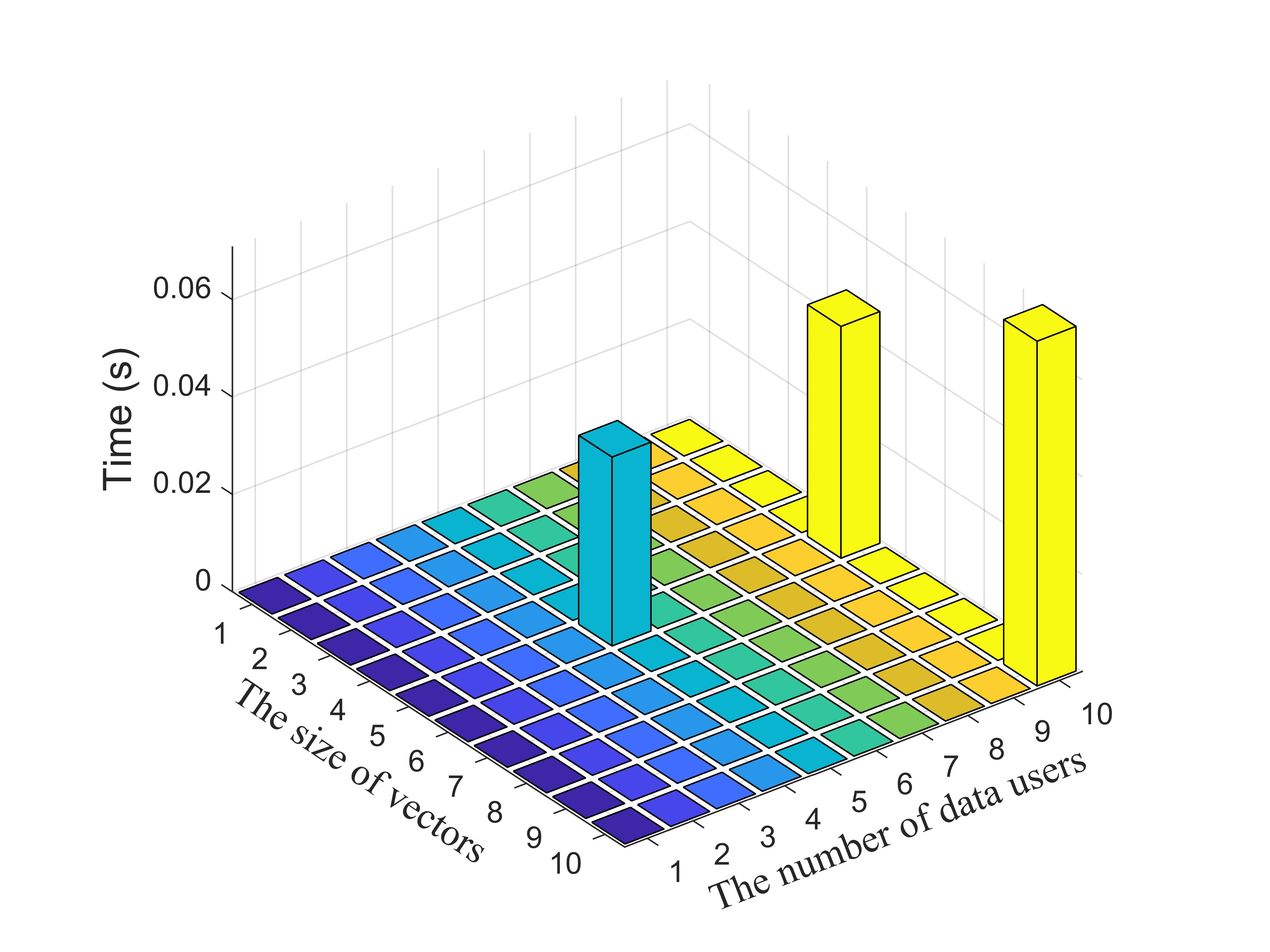}}
        
        \subfloat[$\mathbf{Enc}$]{\label{figure:test_Enc}\includegraphics[width = 0.25\textwidth]{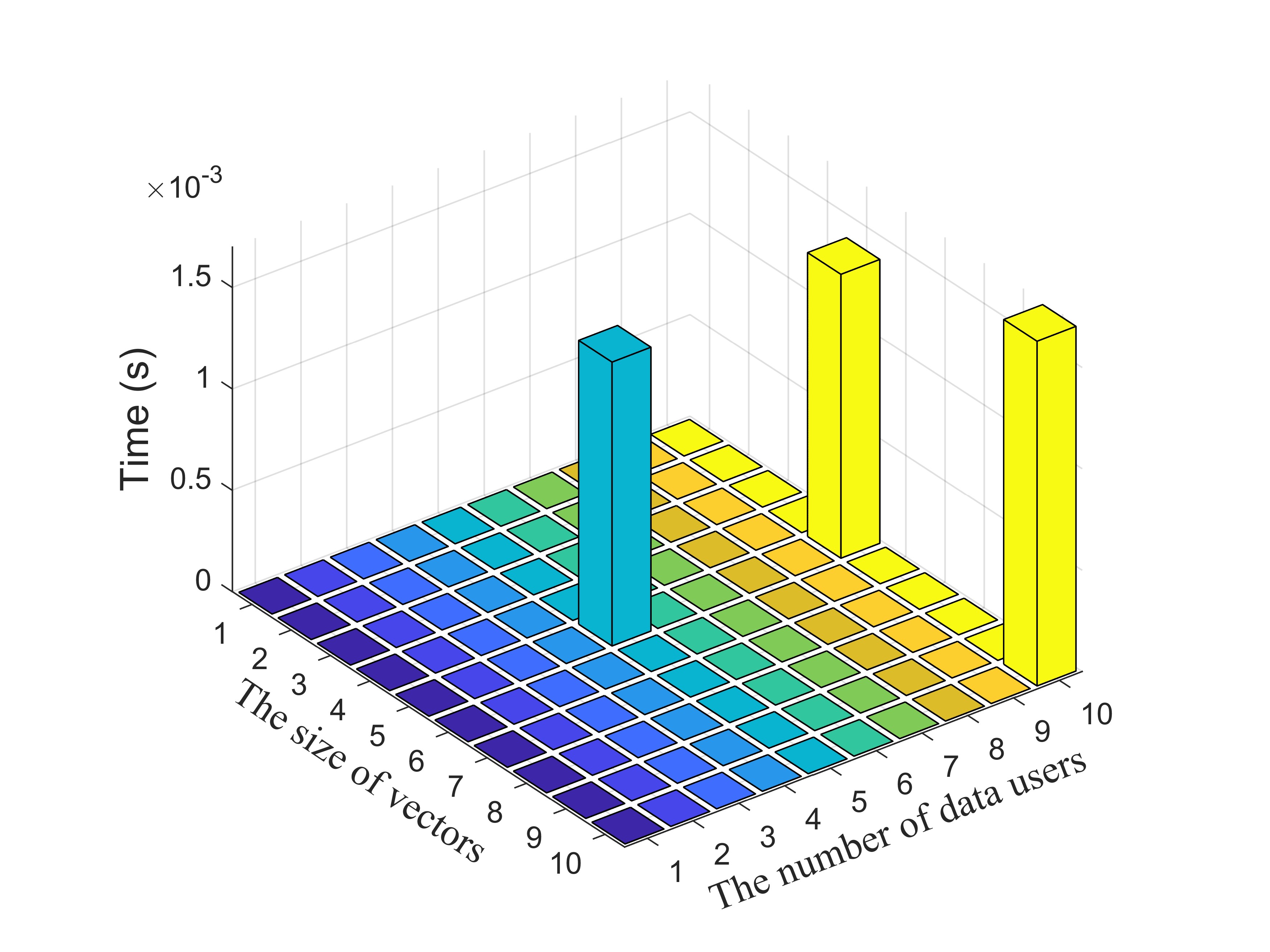}}
        \subfloat[$\mathbf{Dec}$]{\label{figure:test_Dec}\includegraphics[width = 0.25\textwidth]{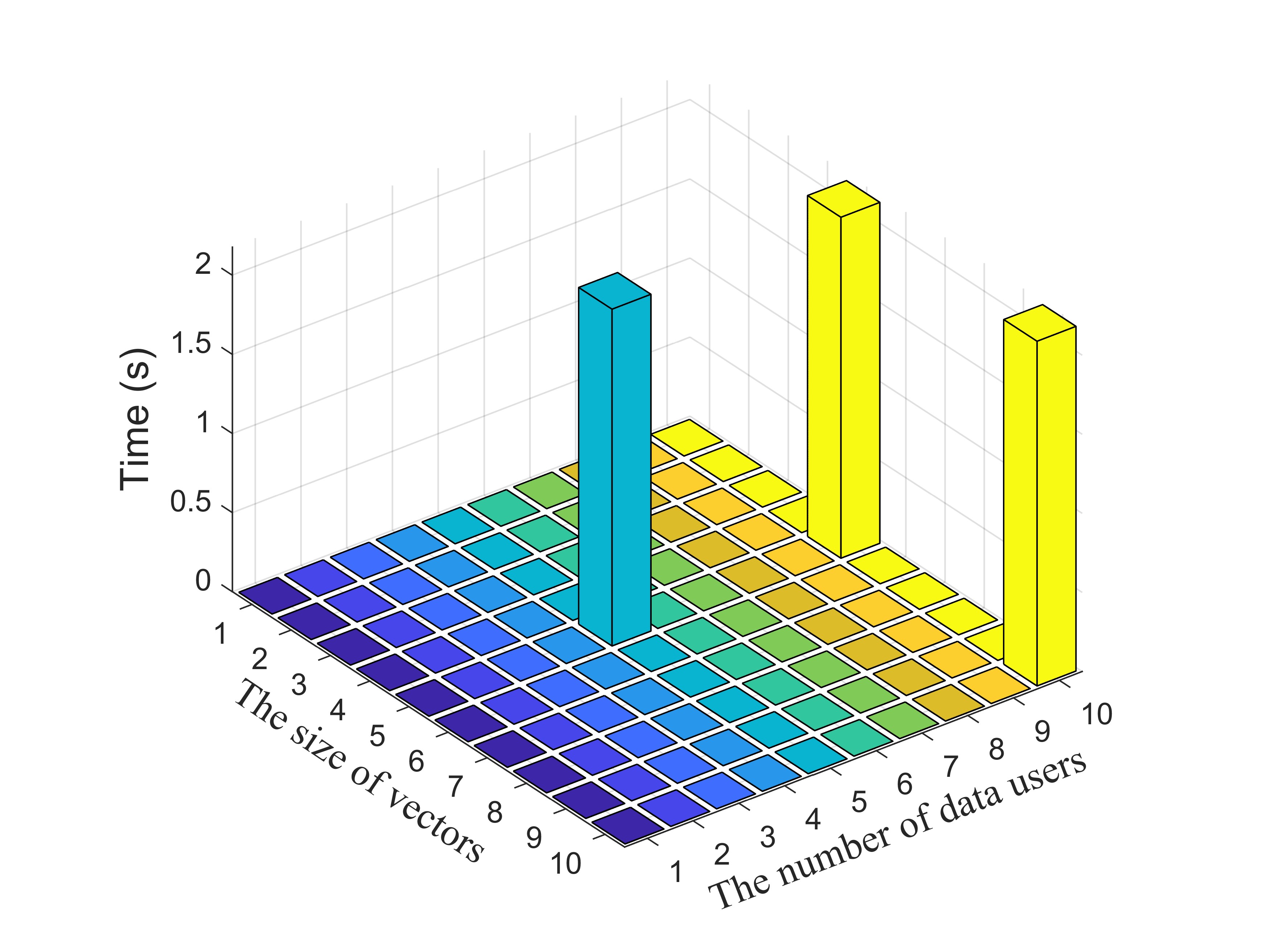}}
        \subfloat[$\mathbf{GroupUpdate}$]{\label{figure:test_GroupUpdate}\includegraphics[width = 0.25\textwidth]{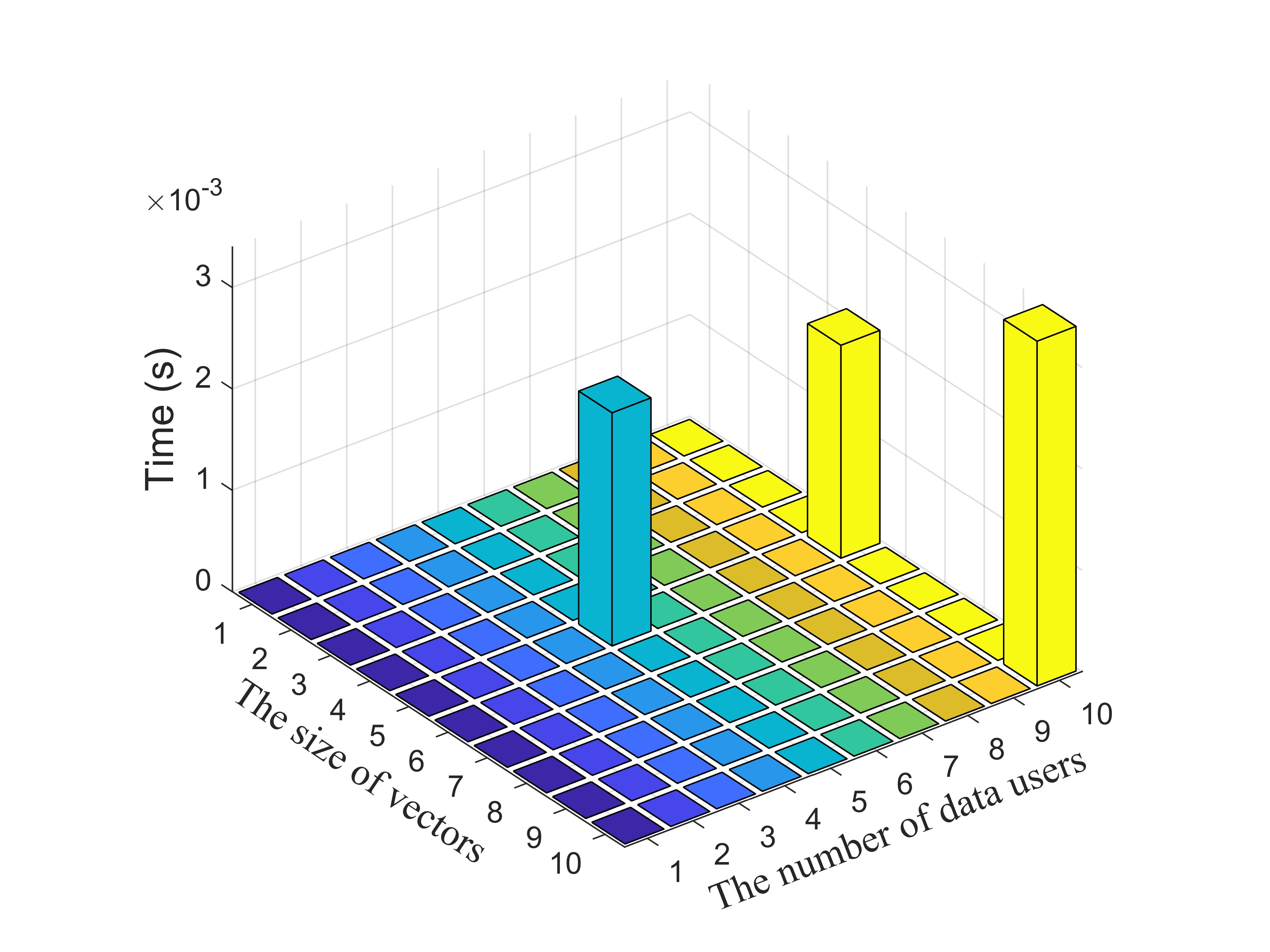}}
        \subfloat[$\mathbf{UptKeyGen}$]{\label{figure:test_RkGen}\includegraphics[width = 0.25\textwidth]{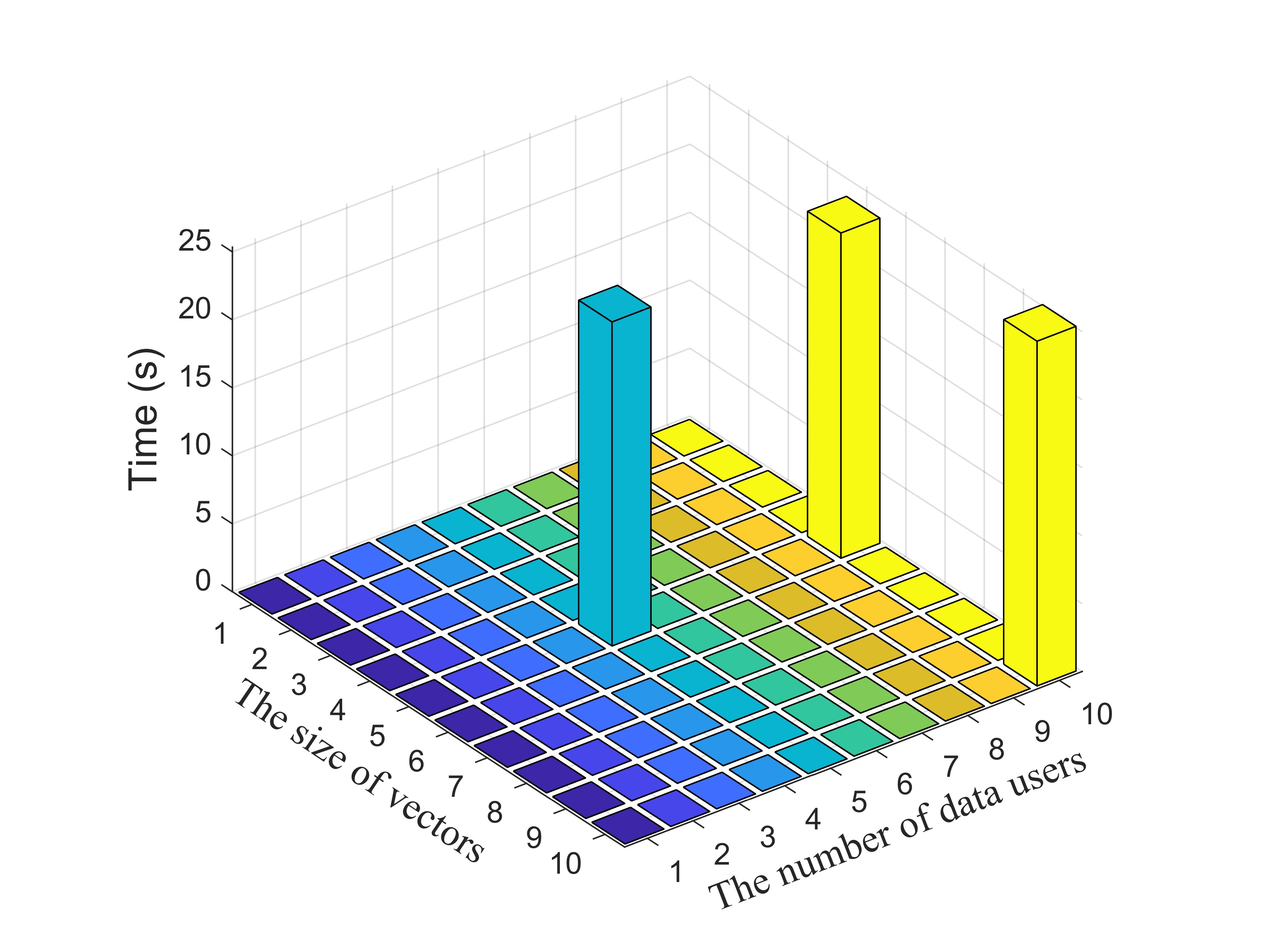}}
        
        \subfloat[$\mathbf{CTUpdate}$]{\label{figure:test_ReEnc}\includegraphics[width = 0.25\textwidth]{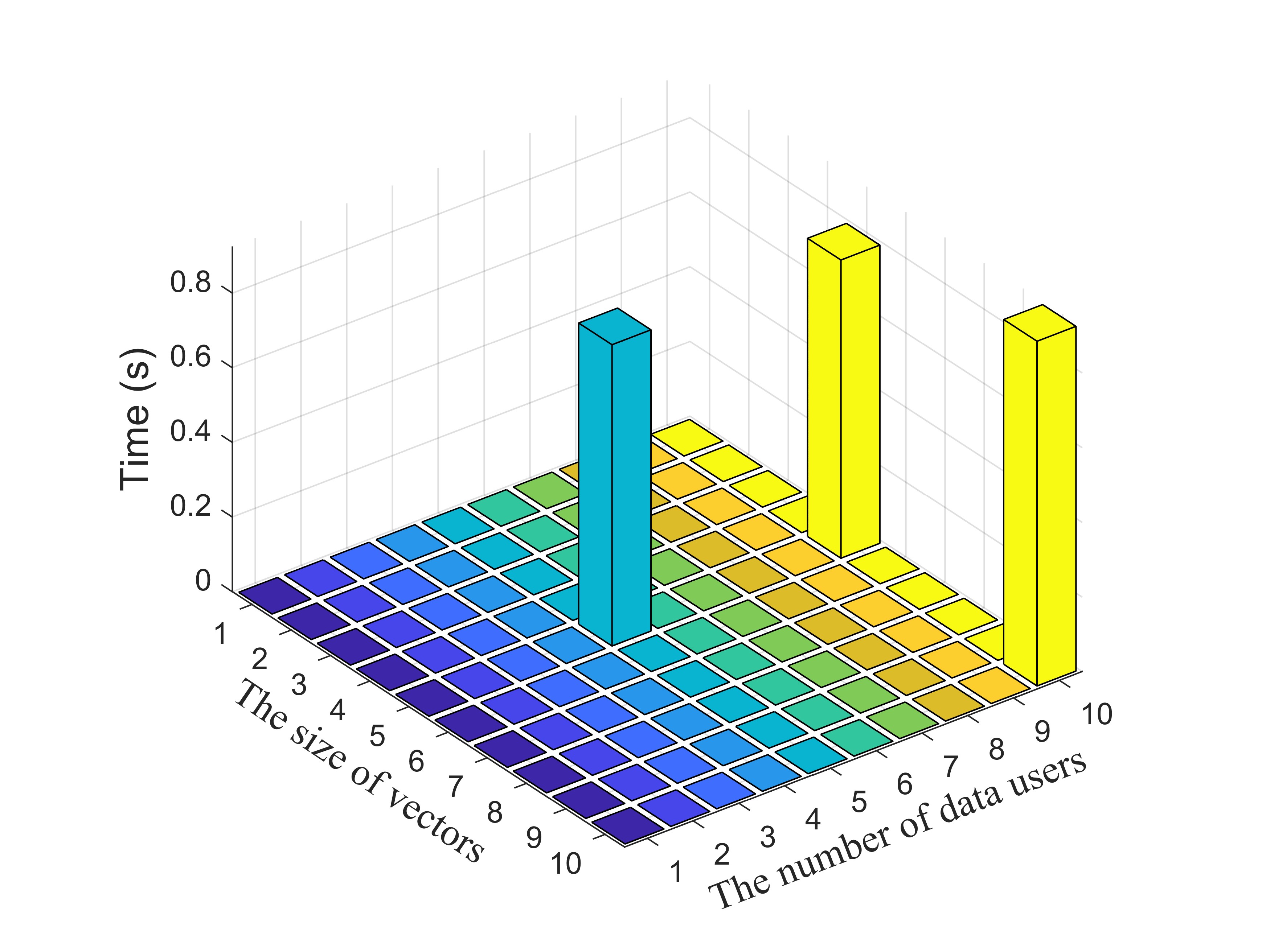}}
        \subfloat[$\mathbf{FUpdate}$]{\label{figure:test_FUpdate}\includegraphics[width = 0.25\textwidth]{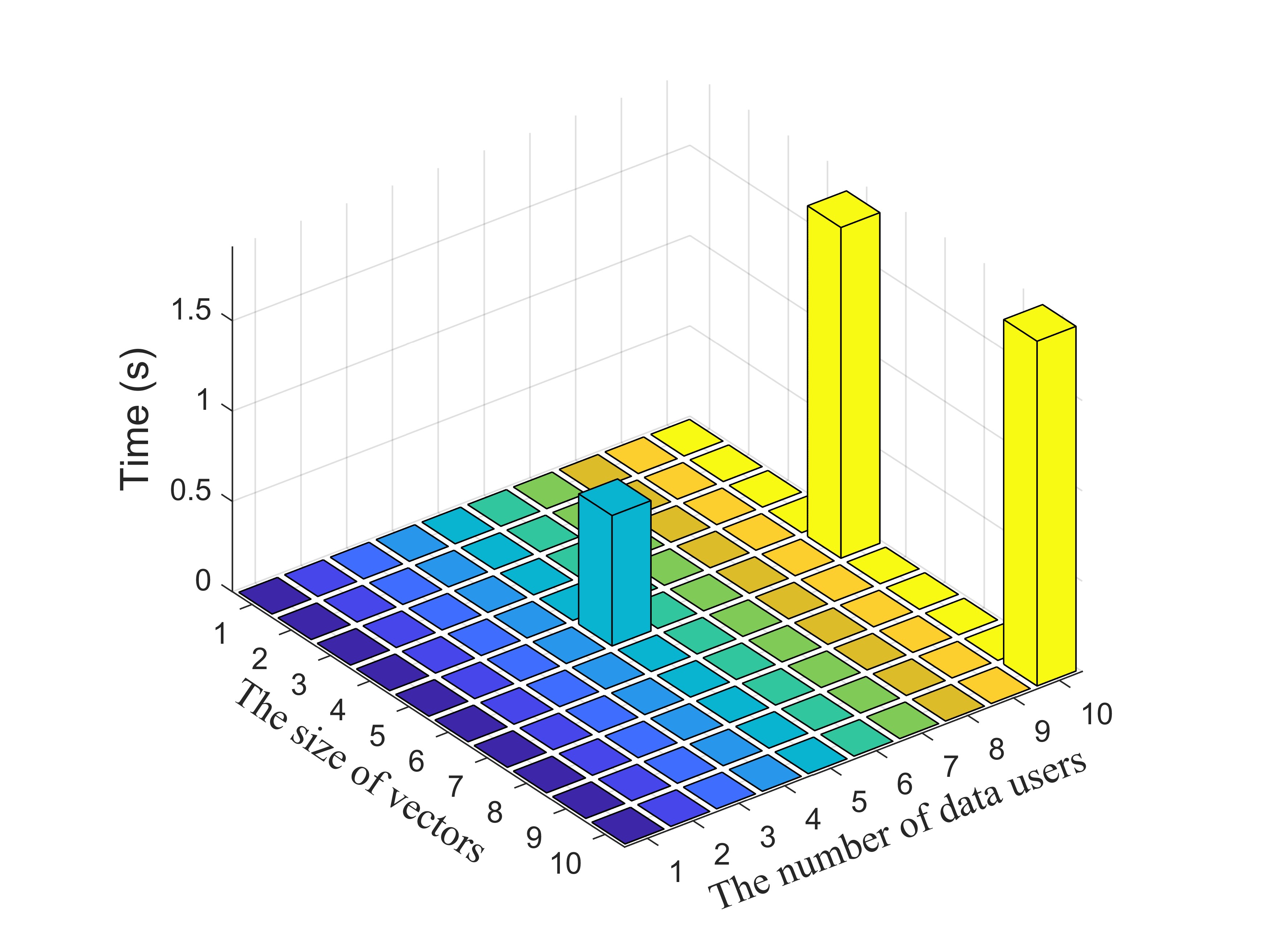}}
        \subfloat[$\mathbf{KeyUpdate}$]{\label{figure:test_KeyUpdate}\includegraphics[width = 0.25\textwidth]{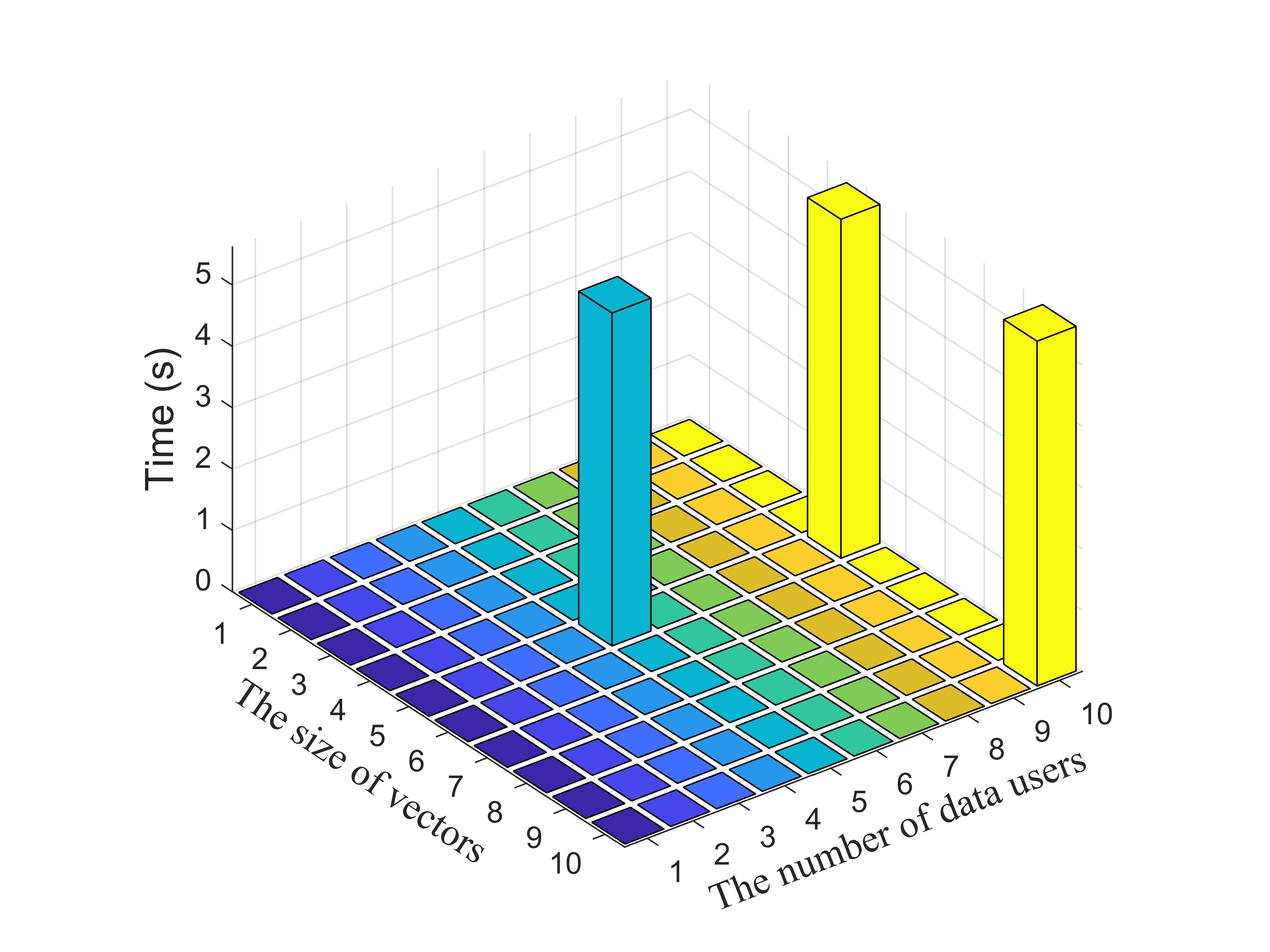}}
        \caption{The computational cost of algorithms in our IPFE-FR scheme}\label{figure:test_time}
\end{figure*}

In our implementation, each algorithm is evaluated using the following three parameter configurations: 
(1) $l_1=5,\mathcal{N}=5$; (2) $l_1=5,\mathcal{N}=10$; (3) $l_1=10,\mathcal{N}=10$, where 
$l_1$ and $\mathcal{N}$  denotes the dimension of the vectors and 
represents the upper bound number of medical institutions, respectively.

As shown in Fig. \subref*{figure:test_SystemSetup}, the $\mathbf{SystemSetup}$ algorithm takes 0.0134 s, 0.0175 s, and 0.0132 s in the three respective cases.
These results indicate that the computational cost of the $\mathbf{SystemSetup}$ algorithm remains constant across different parameter settings. 

Fig. \subref*{figure:test_GroupSetup} illustrates the computational cost of the $\mathbf{GroupSetup}$ algorithm, which takes 0.0049 s, 0.0062 s, and 0.0074 s in the three respective cases.
The results show that the computational cost of $\mathbf{GroupSetup}$ grows linearly with both $l_1$ and $\mathcal{N}$. 

Fig. \subref*{figure:test_UKeyGen} illustrates the computational cost of the $\mathbf{UKeyGen}$ algorithm. 
This algorithm requires 0.0009 s, 0.0016 s, and 0.0015 s to generate user keys for 
$\mathcal{N}$ medical institutions in the three cases, respectively.
This indicates that the computational cost of $\mathbf{UKeyGen}$ grows linearly with 
$\mathcal{N}$.

Fig. \subref*{figure:test_FKeyGen} presents the computational cost of the $\mathbf{FKeyGen}$ algorithm for 
$\mathcal{N}$ medical institutions, which takes 0.0388 s, 0.0476 s, and 0.0708 s in the three cases, respectively.
These results demonstrate that the computational cost of $\mathbf{FKeyGen}$ grows linearly with both 
$l_1$ and $\mathcal{N}$.

Fig. \subref*{figure:test_Enc} illustrates the computational cost of the $\mathbf{Enc}$ algorithm. The algorithm takes 0.0014 s, 0.0014 s and 0.0017 s 
in the three cases, respectively. 
The results show that the $\mathbf{Enc}$ algorithm's computational cost grows linearly with $l_1$.

Fig. \subref*{figure:test_Dec} shows the computational cost of the $\mathbf{Dec}$ algorithm, 
which takes 2.1285 s, 2.1547 s, and 2.1785 s in the three cases, respectively. The computational cost of $\mathbf{Dec}$ 
grows linearly with both $l_1$ and $\mathcal{N}$. 

As shown in Fig. \subref*{figure:test_GroupUpdate}, the $\mathbf{GroupUpdate}$ algorithm takes 0.0023 s, 0.0021 s, and 0.0034 s to update the system version in the three cases, respectively. 
The computational cost of $\mathbf{GroupUpdate}$ grows 
linearly with $l_1$. 

Fig. \subref*{figure:test_RkGen} shows the computational cost of the $\mathbf{UptKeyGen}$ algorithm, which takes 23.8351 s, 23.9209 s, and 25.3554 s in the three cases, respectively. 
The computational cost of the $\mathbf{UptKeyGen}$ algorithm grows linearly with $l_1$.

The runtime of the $\mathbf{CTUpdate}$ for updating ciphertext is illustrated in Fig. \subref*{figure:test_ReEnc}. 
In the three cases, the algorithm takes 0.8071 s, 0.7989 s, and 0.9235 s, respectively. 
This indicates that the computational cost of $\mathbf{CTUpdate}$ grows linearly with $l_1$.

Fig. \subref*{figure:test_FUpdate} shows the computational cost of the $\mathbf{FUpdate}$ algorithm for a single function. The algorithm takes 0.7219 s, 1.8303 s, and 1.908 s in the three cases, respectively. 
These results demonstrate that the computational cost of $\mathbf{FUpdate}$ grows linearly with both $l_1$ and $\mathcal{N}$.

Fig. \subref*{figure:test_KeyUpdate} shows the computational cost of the $\mathbf{KeyUpdate}$ algorithm, 
which takes 5.4285 s, 5.5255 s, and 5.6204 s in the three cases, respectively. 
This indicates that the computational cost of $\mathbf{KeyUpdate}$ grows linearly with both $l_1$ and $\mathcal{N}$.

\section{Security}\label{section:sixth}
The security proof of the IPFE-FR scheme is provided in the following section.
Moreover, we analyse the collusion resistance of our scheme. 

\begin{theorem}
Let $\lambda$ be the security parameter. 
The IPFE-FR scheme is 
$(t,\epsilon(\lambda))$-sIND-CPA secure 
if $\mathsf{PRF}$ is a secure pseudorandom function, ALS-IPFE is $(t,\epsilon'(\lambda))$-IND-CPA secure
and $H_{1}$, $H_{2}$ are random oracles, where 
$\epsilon(\lambda) \leq \epsilon'(\lambda)$. 
\end{theorem}

\begin{proof}
    Suppose that if a PPT adversary $\mathcal{A}$ exists with a non-negligible advantage 
    against the sIND-CPA security of our IPFE-FR scheme, then this implies the existence of an algorithm
    $\mathcal{B}$ that can break the IND-CPA security of 
    the ALS-IPFE scheme \cite{agrawal2016fully}. 
    Notably, the ALS-IPFE scheme is secure 
    under the $\mathrm{LWE}_{n,q,m,\alpha}$ assumption \cite{agrawal2016fully}. 
    Let $\mathcal{S}$ be the challenger in the IND-CPA game for the ALS-IPFE scheme. 

    $\mathbf{Game\ 0:}$ The original sIND-CPA game, as described in Section $\ref{model}$. 
    
    $\mathbf{Game\ 1:}$
    Apart from the modifications listed below, this game is the same as the previous one.
    Rather than generating $(\mathbf{A},\mathbf{T}_{\mathbf{A}})$ using $\mathsf{TrapGen}(1^n, 1^m, p)$, 
    we sample $\mathbf{A}$ uniformly from $\mathbb{Z}_p^{n \times m}$
    Moreover, we sample  
    $\mathbf{B}_{ver} \gets \mathcal{D}_{\mathbb{Z}^{m\times l_1},\rho_1}$ for $ver = 1,\ldots , ver^* -1 $,  
    $\mathbf{C} \gets \mathbb{Z}_p^{n\times l_1}$ and compute 
    $\mathbf{U}_{ver} = \mathbf{A}\cdot\mathbf{B}_{ver} -\mathbf{C}$ for $ver = 1,\ldots , ver^* -1 $. 
    Hence, we can use $\{\mathbf{B}_{ver}\}_{ver \in \left[ 1,\ldots , ver^*-1 \right]}$ to answer 
    function key queries for $ver \neq ver^*$. 

    By Lemma \ref{TrapGen}, we know that $\mathbf{B}_{ver} \gets \mathcal{D}_{\mathbb{Z}^{m\times l_1},\rho_1}$ is statistically close to 
    $\mathbf{Z}_{ver} \gets \mathsf{SamplePre}(\mathbf{A},\mathbf{T}_{\mathbf{A}}, \mathbf{C}+\mathbf{U}_{ver},\rho_1)$, and
    $\mathbf{A}\cdot \mathbf{B}_{ver} \in \mathbb{Z}_p ^{n\times l_1}$ is statistically close to uniform distribution over $\mathbb{Z}_p ^{n\times l_1}$. 
    Since $\mathbf{C}$ is uniformly distributed over $\mathbb{Z}_p ^{n\times l_1}$, 
    $\mathbf{U}_{ver}$ is statistically close to uniform distribution over $\mathbb{Z}_p ^{n\times l_1}$.
    Hence,  $\mathbf{Game\ 1}$ is statistically close to $\mathbf{Game\ 0}$. 
    
    $\mathbf{Game\ 2:}$
    The security of our scheme is reduced to the ALS-IPFE scheme in this game.

    $\mathbf{Init.}$ 
    $\mathcal{A}$ submits a version number $ver^*$, 
    a vector $\mathbf{x}^{*} \in \mathbb{Z}_p^{l_1}$ and 
    a set of revoked users $\mathcal{R}^*$ 
    to $\mathcal{B}$. 

    $\mathbf{Setup.}$ 
    $\mathcal{B}$ obtains $(\mathbf{A}_{ALS},\mathbf{U}_{ALS})$ from $\mathcal{S}$.
    Then, $\mathcal{B}$ sets $\mathbf{A} = \mathbf{A}_{ALS}$,  
    samples $\mathbf{V}\gets \mathbb{Z}^{n\times m}_q$, $\mathbf{D} \gets \mathcal{D}_{\mathbb{Z}^{m\times l_2},\rho_2}$, 
    $\mathbf{B}_{ver} \gets \mathcal{D}_{\mathbb{Z}^{m\times l_1},\rho_1}$ for $ver = 1,\ldots , ver^*-1$ and 
    computes $\mathbf{F}=\mathbf{V}\cdot \mathbf{D} \in \mathbb{Z}_q^{n\times l_2}$. 
    
    Furthermore, $\mathcal{B}$ computes
    $\mathbf{U}_{ver} = \mathbf{A}\cdot\mathbf{B}_{ver} -\mathbf{C}$ 
    for $ver = 1,\ldots, ver^*-1$
    and sets 
    $\mathbf{U}_{ver^*} = \mathbf{U}_{ALS} -\mathbf{C}$, which implies 
    $\mathbf{Z}_{ver} = \mathbf{B}_{ver}$ for $ver = 1,\ldots, ver^*$,  
    $\mathbf{Z}_{ver^*} = \mathbf{Z}_{ALS}$. 
    Let 
       $H_1: \{0,1\}^* \to \mathbb{Z}_p^{l_2}$, $H_2: \mathbb{Z}_p \to \{ 0,1\} ^ t $ be cryptographic hash functions 
       and $\mathsf{PRF}:\mathbb{Z}_p \times \mathbb{Z}_p^{l_1} \to \mathbb{Z}_p^{l_2}$ is a 
       pseudorandom function, where 
       $t = m \cdot\left\lceil \mathrm{log\ }(2Xl_1\rho_2) \right\rceil$ 

    To generate update key, $\mathcal{B}$  
    samples $\mathbf{E}_{1} \gets \mathbb{Z}_{p}^{hm \times n}$, 
    $\mathbf{E}_{2} \gets \mathcal{D}_{\mathbb{Z}^{hm\times m},\sigma_1}$ and 
    $\mathbf{E}_{3} \gets \mathcal{D}_{\mathbb{Z}^{hm\times l_1},\sigma_1}$ where $h=\lceil \mathrm{log\ }p \rceil$, 
    and computes 
    $\mathsf{uptk}_{ver+1}=$
    \begin{equation*}
    \begin{bmatrix} 
    \mathbf{E}_{1}\mathbf{A}_1+\mathbf{E}_{2} & \ \ \mathbf{E}_{1} \cdot (\mathbf{C}+ \mathbf{U}_{ver+1}) + \mathbf{E}_{3}-\mathsf{PowerT}_{p}(\mathbf{Z}_{ver})\\ 
    \mathbf{0}_{l_1\times m} & \ \ \mathbf{I}_{l_1\times l_1} \\
    \end{bmatrix} 
    \end{equation*}
    for $ver = 1,\dots ver^*-1$. 

    Finally, $\mathcal{B}$ stores $\{\mathbf{D}, \mathbf{Z}_{ver},
    \mathsf{uptk}_{ver+1}\}_{ver\in \left[1,ver^*-1\right]}$ and sends 
    $(\mathsf{mpk}= (\mathbf{A},\mathbf{V},\mathbf{C}),$ 
    $\left\{ \mathsf{gpk}_{ver}=(\mathbf{F},\mathbf{U}_{ver})\right\}_{ver\in \left[1,ver^*\right]},$
    $\mathsf{pp})$ to $\mathcal{A}$. 

    $\mathbf{Phase-1.}$

    $\mathsf{User\ Key\ Query}.$ $\mathcal{A}$ adaptively submits an identity $\mathsf{id}$ 
    to $\mathcal{B}$. 
    $\mathcal{B}$ 
    sets $\mathbf{id} = H_1(\mathsf{id}) \in \mathbb{Z}^{l_2}_p$, computes $\mathbf{u}_{\mathsf{id}}=\mathbf{D}\cdot \mathbf{id}$, 
    and forwards $\mathsf{usk}_{\mathsf{id}} = (\mathbf{id},\mathbf{u}_{\mathsf{id}})$ 
    to $\mathcal{A}$. 
    
    $\mathsf{Function\ Key\ Query.}$ 
    $\mathcal{A}$ adaptively submits $(\mathbf{x},\mathsf{id},ver)$ 
    to $\mathcal{B}$, where $\mathbf{x}\in \mathbb{Z}_p^{l_1}$ and $ver  \in \left[ 1,ver^*\right]$, 
    with the restriction that $(\mathbf{x},ver) \neq (\mathbf{x}^*,ver^*)$. 
    If $ver=ver^*$, $\mathcal{B}$ forwards $\mathbf{x}$ to the ALS-IPFE challenger $\mathcal{S}$ and obtains 
    $(\mathbf{x},\mathsf{sk}_{\mathbf{x}}^{ALS})$. Then, $\mathcal{B}$ lets $\mathbf{f}_{\mathbf{x},ver} = 
    \mathsf{sk}^{ALS}_{\mathbf{x}}$; otherwise, $\mathcal{B}$ 
    computes $\mathbf{f}_{\mathbf{x},ver} = \mathbf{Z}_{ver} \cdot \mathbf{x}$. 
    Then, 
    $\mathcal{B}$ sets 
    $\mathbf{t}_{\mathbf{x}} = \mathsf{PRF}(k_{p},\mathbf{x}) \in \mathbb{Z}^{l_2}_p$, 
    $\mathbf{id} = H_1(\mathsf{id}) \in \mathbb{Z}^{l_2}_p$
    and computes $\mathbf{f}_{\mathbf{x},\mathsf{id},ver} = \mathbf{f}_{\mathbf{x},ver} - 
    (\overbrace{ v_{\mathbf{x},\mathsf{id}},\cdots, v_{\mathbf{x},\mathsf{id}}}^{m})^{\top} \in \mathbb{Z}^m$, where 
    $v_{\mathbf{x},\mathsf{id}} = \langle \mathbf{id}, \mathbf{t}_{\mathbf{x}}\rangle$ mod $p$.
    
    Let the public directory $\mathsf{pd}$ contains $(pd_{\mathbf{x},1},pd_{\mathbf{x},2})$, 
    where $\mathbf{x}$ is key queries that have been made so far. 
    If $(pd_{\mathbf{x},1},pd_{\mathbf{x},2}) \notin \mathsf{pd}$,  
    sample $\mathbf{s}_1 \gets \mathbb{Z}_{q}^{n}$, 
    $\mathbf{e}_{1} \gets \mathcal{D}_{\mathbb{Z},\sigma_2}^{m}$, 
    $\mathbf{e}_{2} \gets \mathcal{D}_{\mathbb{Z},\sigma_2}^{l_2}$ and 
    computes $pd_{\mathbf{x},1} = \mathbf{V}^{\top} \cdot \mathbf{s}_1 + \mathbf{e}_{1}$ and 
    $pd_{\mathbf{x},2} = \mathbf{F}^{\top} \cdot \mathbf{s}_1 + \mathbf{e}_{2} + p^{k-1} \cdot \mathbf{t}_{\mathbf{x}}$. 
    Then, append $(pd_{\mathbf{x},1},pd_{\mathbf{x},2})$ to $\mathsf{pd}$.
    
    Lastly, $\mathcal{B}$ sends $\mathsf{fsk}_{\mathbf{x},\mathsf{id},ver} = (\mathbf{x},\mathbf{f}_{\mathbf{x},\mathsf{id},ver})$ 
    to $\mathcal{A}$ and 
    adds $\mathbf{x}$ into a (initially empty) set $T_x$ if $\mathbf{x} \notin T_x$. 

    $\mathsf{Update\ Key\ Query.}$ $\mathcal{A}$ adaptively submits a version number $ver \in [1,ver^{*}-1]$ 
    to $\mathcal{B}$. $\mathcal{B}$ sends $\mathsf{uptk}_{ver+1}$ to $\mathcal{A}$. Notably, $\mathsf{uptk}_{ver+1}$ has been 
    generated by $\mathcal{B}$ in $\mathbf{Setup}$ phase. 

    $\mathsf{Function\ Update\ Query.}$ $\mathcal{A}$ submits a version number $ver \in [1,ver^{*}-1]$, a vector $\mathbf{x}\in T_{\mathbf{x}}$ 
    and a set of revoked users $\mathcal{R}_{\mathbf{x}}$ for $\mathbf{x}$, with the restriction that 
    $\mathcal{R}^{*} \subseteq \mathcal{R}_{\mathbf{x}}$ if $(\mathbf{x},ver) = (\mathbf{x}^*,ver^*-1)$.   
    $\mathcal{B}$ randomly selects $K\in \mathbb{Z}_{p}$, 
    computes $\mathbf{v}_{\mathbf{x}}\in \mathbb{Z}^{l_2}_{p}\backslash \{\mathbf{0}\}$ such that 
    $\langle \mathbf{id},\mathbf{v}_{\mathbf{x}} \rangle = 0$ for every $\mathsf{id} \in  \mathcal{R}_{\mathbf{x}}$ 
    and samples 
    $\mathbf{s}_3 \gets \mathbb{Z}_{q}^{n}$, $\mathbf{e}_{5} \gets \mathcal{D}_{\mathbb{Z}^{m},\sigma_1}$ and
    $\mathbf{e}_{6} \gets \mathcal{D}_{\mathbb{Z}^{l_{1}},\sigma_1}$. 
    \begin{itemize}
        \item Case 1: $ver \neq ver^*-1$. 
            $\mathcal{B}$ computes 
            $\mathsf{upi}_{\mathbf{x},ver,1}=\mathbf{V}^{\top} \cdot \mathbf{s}_3 + \mathbf{e}_{5}$, 
            $\mathsf{upi}_{\mathbf{x},ver,2}=\mathbf{F}^{\top} \cdot \mathbf{s}_3 + \mathbf{e}_{6} + p^{k-1} \cdot (k_t \cdot \mathbf{v}_{\mathbf{x}})$ and 
            $\mathsf{upi}_{\mathbf{x},ver,3} = 
            H_2(k_t) \oplus \mathsf{bin}((\mathbf{B}_{ver+1} - \mathbf{B}_{ver})\cdot \mathbf{x})$. 
        \item Case 2: $ver=ver^*-1$ and $\mathbf{x} \neq \mathbf{x}^*$. 
            $\mathcal{B}$ sends $\mathbf{x}$ to the ALS-IPFE challenger $\mathcal{S}$ and obtains a secret key 
            $(\mathbf{x},\mathsf{sk}^{ALS}_{\mathbf{x}})$ for the vector $\mathbf{x}$. 
            $\mathcal{B}$ computes 
            $\mathsf{upi}_{\mathbf{x},ver,1}=\mathbf{V}^{\top} \cdot \mathbf{s}_3 + \mathbf{e}_{5}$, 
            $\mathsf{upi}_{\mathbf{x},ver,2}=\mathbf{F}^{\top} \cdot \mathbf{s}_3 + \mathbf{e}_{6} + p^{k-1} \cdot (k_t \cdot \mathbf{v}_{\mathbf{x}})$ and 
            $\mathsf{upi}_{\mathbf{x},ver,3} = 
            H_2(k_t) \oplus \mathsf{bin}(\mathsf{sk}^{ALS}_{\mathbf{x}}- \mathbf{B}_{ver}\cdot \mathbf{x})$. 
        \item Case 3: $(\mathbf{x},ver) = (\mathbf{x}^*,ver^*-1)$. 
            $\mathcal{B}$ randomly selects $rb \gets \{0,1\}^t$, 
            and computes 
            $\mathsf{upi}_{\mathbf{x},ver,1}=\mathbf{V}^{\top} \cdot \mathbf{s}_3 + \mathbf{e}_{5}$, 
            $\mathsf{upi}_{\mathbf{x},ver,2}=\mathbf{F}^{\top} \cdot \mathbf{s}_3 + \mathbf{e}_{6} + p^{k-1} \cdot (k_t \cdot \mathbf{v}_{\mathbf{x}})$ and 
            $\mathsf{upi}_{\mathbf{x},ver,3} = 
            H_2(k_t) \oplus rb $. 
    \end{itemize}
    $\mathcal{B}$ forwards $\mathsf{UPI}_{\mathbf{x},ver}=(\mathsf{upi}_{\mathbf{x},ver,1},\mathsf{upi}_{\mathbf{x},ver,2},$ $\mathsf{upi}_{\mathbf{x},ver,3},\mathbf{v}_{\mathbf{x}})$
    to $\mathcal{A}$. 

    $\mathbf{Challenge.}$ $\mathcal{A}$ submits two vectors $\mathbf{y}_{0}$ and $\mathbf{y}_{1}$,   
constrained such that for all $\mathbf{x} \in T_x \setminus \{\mathbf{x}^*\}$, the equality $\langle \mathbf{x}, \mathbf{y}_{0} \rangle = \langle \mathbf{x}, \mathbf{y}_{1} \rangle$ holds.
    $\mathcal{B}$ forwards $(\mathbf{y}_{0},\mathbf{y}_{1})$ to the ALS-IPFE challenger $\mathcal{S}$ and 
    obtains the challenger ciphertext $\mathsf{ct}_{\mathbf{y}_b} =(\mathbf{c}^{ALS}_1,\mathbf{c}^{ALS}_2)$. 
    Then, $\mathcal{B}$  returns $\mathsf{ct}_{\mathbf{y}_b} =(\mathbf{c}^{ALS}_1,\mathbf{c}^{ALS}_2)$ to $\mathcal{A}$. 

    $\mathbf{Phase\ 2.}$  $\mathcal{A}$ is allowed to make queries  
    for user keys, function keys, update keys, and function update information, all of which are answered by
    $\mathcal{B}$ following the procedure defined in $\mathbf{Phase\ 1}$. 
        
    $\mathbf{Guess.}$ $\mathcal{A}$ produces a guess $b'\in\{0,1\}$ for the bit $b$. Then, $\mathcal{B}$ returns $b'$ to the ALS-IPFE 
    challenger $\mathcal{S}$ as its guess.  
    If $b=b'$, then $\mathcal{A}$ wins the above game, and as a result,  $\mathcal{B}$ 
    successfully breaks the IND-CPA security of the ALS-IPFE scheme.
    To sum up, $\epsilon'(\lambda)$ represents the maximum advantage of all PPT adversary in compromising the IND-CPA security of the ALS-IPFE scheme. 
    Therefore, the advantage $\epsilon(\lambda)$ of all PPT adversaries $\mathcal{A}$ in winning the above this game satisfies 
    $\epsilon(\lambda) \leq  \epsilon'(\lambda)$. 
\end{proof} 

To capture the forward security of our scheme, we allow the adversary $\mathcal{A}$ to obtain the function keys 
$\mathsf{fsk}_{\mathbf{x}^*,\mathsf{id},ver}$, where $ver\in \{1,\ldots,ver^*-1\}$. 

$\mathbf{Collusion\ resistance}.$ In our IPFE-FR scheme, 
$MI$s cannot collude with each other to perform inner product calculations over encrypted data.
Without loss of generality, we take medical institutions $MI$ and $MI'$ as an example. 
The $MI$ and $MI'$ have secret keys $(\mathsf{usk}_{\mathsf{id}}, 
\mathsf{fsk}_{\mathbf{x},\mathsf{id},ver})$ and 
$(\mathsf{usk}_{\mathsf{id}'}, \mathsf{fsk}_{\mathbf{x}',\mathsf{id}',ver})$, respectively. 
According to our scheme, we have 
$\mathsf{usk}_{\mathsf{id}} = (\mathbf{id}, \mathbf{u}_{\mathsf{id}}= \mathbf{D} \cdot \mathbf{id})$, 
$\mathsf{fsk}_{\mathbf{x},\mathsf{id},ver} = (\mathbf{x},\mathbf{f}_{\mathbf{x},\mathsf{id},ver} = \mathbf{Z}_{ver} \cdot \mathbf{x} - 
(\overbrace{v_{\mathbf{x},\mathsf{id}}, \cdots, v_{\mathbf{x},\mathsf{id}} }^{m})^{\top})$, 
$\mathsf{usk}_{\mathsf{id}'} = (\mathbf{id}', \mathbf{u}_{\mathbf{x}',\mathsf{id}'}= \mathbf{D} \cdot \mathbf{id}')$ 
and $\mathsf{fsk}_{\mathbf{x}',\mathsf{id}',ver} = (\mathbf{x}',\mathbf{f}_{\mathbf{x}',\mathsf{id}',ver} = \mathbf{Z}_{ver} \cdot \mathbf{x}'
- (\overbrace{v_{\mathbf{x}',\mathsf{id}'},
\cdots,v_{\mathbf{x}',\mathsf{id}'}}^{m})^{\top})$.
Notably, the ciphertext $(pd_{\mathbf{x},1}, pd_{\mathbf{x},2})$ encrypting $\mathbf{t}_{\mathbf{x}}$ and 
the ciphertext $(pd_{\mathbf{x}',1},pd_{\mathbf{x}',2})$ encrypting $\mathbf{t}_{\mathbf{x}'}$
are stored in the public directory $\mathsf{pd}$.

For the same $MI$'s secret keys such as $(\mathsf{usk}_{\mathsf{id}}, 
\mathsf{fsk}_{\mathbf{x},\mathsf{id},ver})$, 
one can use $\mathsf{usk}_{\mathsf{id}}$ to learn $v_{\mathbf{x},\mathsf{id}}$ 
from $(pd_{\mathbf{x},1},pd_{\mathbf{x},2})$. 
Then, he/she can further compute $\mathbf{f}_{\mathbf{x},ver} = 
\mathbf{Z}_{ver}\cdot \mathbf{x} = \mathbf{f}_{\mathbf{x},\mathsf{id},ver} + (\overbrace{v_{\mathbf{x},\mathsf{id}},
\cdots, v_{\mathbf{x},\mathsf{id}}}^{m})^{\top}
= \mathbf{Z}_{ver}\cdot \mathbf{x} + (\overbrace{v_{\mathbf{x},\mathsf{id}},
\cdots, v_{\mathbf{x},\mathsf{id}}}^{m})^{\top} -(\overbrace{v_{\mathbf{x},\mathsf{id}},
\cdots, v_{\mathbf{x},\mathsf{id}}}^{m})^{\top}$. 
Hence, the $MI$ can use $\mathbf{f}_{\mathbf{x},ver}$ to compute 
inner product from $\mathsf{CT}_{ver+1}$. 

However, when the secret keys are from different $MI$s, 
such as $(\mathsf{usk}_{\mathbf{x}},\mathsf{fsk}_{\mathbf{x}',\mathsf{id}',ver})$, 
they cannot be combined to perform inner product calculations over encrypted data.. 
Using $\mathsf{usk}_{\mathsf{id}}$, one can only calculate $v_{\mathbf{x},\mathsf{id}}$ and not 
$v_{\mathbf{x}',\mathsf{id}'}$, which is guaranteed by the security of IPFE. 
Hence, the vector $(\overbrace{v_{\mathbf{x}',\mathsf{id}'},
\cdots,v_{\mathbf{x}',\mathsf{id}'}}^{m})^{\top}$ in $\mathbf{f}_{\mathbf{x}',\mathsf{id}',ver}$ cannot be eliminated 
to obtain $\mathbf{f}_{\mathbf{x}',ver} = \mathbf{Z}_{ver} \cdot \mathbf{x}'$ and further, the inner product on encrypted data cannot be calculated. 

\section{Conclusion and Future Work} \label{section:eighth}
In this paper, to achieve fine-grained control over
users’ function computing right in EHR sharing system, we proposed a new notion called IPFE-FR, 
where the revocation of partial function computing rights of medical institutions is supported. 
Moreover, if a medical institution’s function computing right
is revoked, it cannot compute the function of encrypted
EHR data generated before revocation. 
To resist collusion attacks, the secret keys belonging to the same medical institution are tied together. 
We formalised the formal definition and security model of the proposed IPFE-FR.
Furthermore, we proposed a concrete construction of IPFE-FR and proved its security under the LWE assumption, which can resist quantum computing attacks. 
Finally, we conducted theoretical analysis and experimental implementation to evaluate its efficiency. In future work, we aim to construct IPFE-FR scheme based on the Ring Learning with Errors (RLWE) assumption to enhance the efficiency of our current scheme.

\section*{Acknowledgments}
This work was supported by the National Natural Science Foundation of China (Grant No. 62372103) and the Natural Science Foundation of Jiangsu Province (Grant No. BK20231149), the Jiangsu Provincial Scientific Research Center of Applied Mathematics (Grant No.BK202330020) and the Start-up Research Fund of Southeast University (Grant No. RF1028623300).

\bibliographystyle{IEEEtran}
\bibliography{paper}

\end{document}